\documentclass[review, nopreprintline, authoryear]{elsarticle}
\usepackage[a4paper, left=2.5cm, right=2.5cm, top=2.5cm, bottom=2.5cm]{geometry}

\usepackage[english]{babel}
\usepackage[utf8]{inputenc}
\usepackage[T1]{fontenc}
\usepackage{color}
\usepackage{lmodern, microtype}

\usepackage{amsmath}
\usepackage{amsthm}
\usepackage{amstext}
\usepackage{amssymb}
\usepackage{latexsym}
\usepackage{dsfont}
\usepackage{mathrsfs}
\usepackage{setspace}
\usepackage[normalem]{ulem}
\usepackage{bbold}
\usepackage{natbib}
\setcitestyle{authoryear,open={(},close={)}} %Citation-related commands

\DeclareMathAlphabet{\mathbb}{U}{msb}{m}{n}

\usepackage{hyperref}
\hypersetup{
    colorlinks=true,
    unicode=true
}

\usepackage{graphicx}
\usepackage{float}
\usepackage{caption, booktabs}
\usepackage{graphics}
\usepackage{subcaption}
\usepackage{lipsum}
\usepackage{amsmath,amssymb,amsthm,amsfonts}
\usepackage[table, dvipsnames]{xcolor}
\usepackage{mathrsfs}
\usepackage{soul}

\usepackage{threeparttable}

\newtheorem{prop}{Proposition}

\theoremstyle{definition}
\newtheorem{definition}{Definition}

\DeclareMathOperator*{\argmin}{arg\,min}

\newcolumntype{H}{>{\setbox0=\hbox\bgroup}c<{\egroup}@{}}

\newcommand{\SP}[1]{\color{black}{{#1}} \color{black}} 
\newcommand{\SPP}[1]{\color{black}{{#1}} \color{black}}

\usepackage{tikz}
\usepackage{lineno}
\modulolinenumbers[0]

\begin{document}

% Comment out to disable line numbers
%\linenumbers

\begin{frontmatter}

    \journal{Journal of Econometrics
    }

    \title{Moving Aggregate Modified Autoregressive Copula-Based Time Series Models (MAGMAR-Copulas)}


    \author[1]{Sven Pappert
    %\footnote{\noindent The author wants to express his gratitude for helpful discussions at the European Meeting of \hspace*{0.45cm} Statisticians 2023 by the Bernoulli Society and the Young Scientists Workshop 2023 by the \hspace*{0.45cm} German Statistical Society}
    \corref{cor1}}
    \ead{pappert@statistik.tu-dortmund.de}
    \cortext[cor1]{Corresponding author}
\address[1]{Chair of Econometrics, Department of Statistics, TU Dortmund University, Germany}

    \begin{abstract}
   \SPP{Copula-based time series models can model univariate and stationary time series in a flexible way by decomposing the joint distribution of consecutive observations into a copula and the stationary distribution. Implicitly this approach} assumes a finite Markov order. In reality a time series may not follow the Markov property. We modify the copula-based time series models by introducing a moving aggregate (MAG) part into the model updating equation. The functional form of the MAG-part is given as the conditional quantile function corresponding to a copula. The resulting MAG-modified Autoregressive Copula-Based Time Series model (MAGMAR-Copula) is discussed in detail and distributional properties are derived in a D-vine framework. \SPP{We show that the stationary distribution implied by the model is not standard-uniform. Hence we propose an adjustment transformation that recovers the desired standard-uniformity.} The model nests the classical ARMA model and can be interpreted as a non-linear generalization of the ARMA-model. The modeling performance is evaluated by modeling US inflation. Our model is competitive with benchmark models in terms of information criteria. 
    \end{abstract}
    \begin{keyword}
        Copula \sep Copula-based time series  \sep Dependence Modeling \sep Non-linear time series \sep Persistency \sep Vine Copula
    \end{keyword}
\end{frontmatter}

\section{Introduction}
\label{Section:Introduction}
\noindent Copula-based time series models are a class of \SPP{univariate} probabilistic autoregressive time series models which allow for flexible temporal dependence modeling as well as flexible marginal modeling.
The idea amounts to modeling (possibly non-linear) autoregressive temporal dependencies by a copula while modeling marginal properties independently. Given a univariate, strictly stationary time series \SPP{$\{X_t\}_{t\in\mathbb{Z}}$,} the model decomposes the joint distribution, $F_{X_t,\hdots,X_{t-p}}$ of $(p+1)$ consecutive random variables of the time series into their stationary marginal distribution, $F_X := F_{X_t} = F_{X_{t-1}} = \hdots = F_{X_{t-p}}$, and their copula $C = F_{F_{X}(X_t), \hdots, F_{X}(X_{t-p})}$. The marginal distribution, $F_X$, accounts for marginal properties such as mean, variance, skewness etc., while the copula accounts for the temporal dependence structure. The general idea was formulated by \cite{darsow1992copulas} for Markov chains and further investigated in (amongst others) \cite{ibragimov2005copula}. A short review is given in \cite{patton2009copula}. \cite{chen2006estimation} and \cite{beare2010copulas} developed the idea for time series. \cite{simard2015forecasting} explore multivariate copula-based time series forecasting with standard copulas while \cite{beare2015vine}, \cite{brechmann2015copar} and \cite{nagler2022stationary} explore multivariate time series modeling using vine-copulas. \\
% Examples
These models can be used to model non-linearities in temporal dependencies such as asymmetry; see \cite{muller2022copula} for irradiation data and \cite{mcneil2022time} for US inflation. For modeling of heavy-tailedness in the temporal dependence, see \cite{loaiza2018time} and \cite{pappert2023forecasting}. Linear temporal dependencies may also be modeled simply by employing the normal copula.
%When choosing the copula to be the normal copula and the marginal distribution to be the normal distribution, the linear AR model is recovered \citep[see e.g.][]{chen2006estimation, smith2010modeling}.
Thus the model can be understood as a generalization to the AR model, which permits broader temporal dependencies and marginal distributions. \\
%TODO less focus on McNeal & BLadt -> DONE
One caveat of the copula-based time series models is that, implicitly, the process is ascribed to be a Markov process of order $p$, i.e. the true data-generating process is of finite autoregressive order $p$. This means that the variable $X_t$ may depend on a finite number of past values $X_{t-1}, \hdots, X_{t-p}$ only. However, in reality the Markov assumption may fail depending on the time series. The time series that is to be modeled may display persistent behavior i.e. non cut-off autodependence. A popular example of this behavior is volatility clustering. There, persistency is apparent in the second moment of the time series. In principle it is possible to extend the copula-based time series model to arbitrary Markov order, however it becomes unfeasible as the number of parameters increases indefinitely. Hence, we think that there is a need for a copula-based time series model with infinite order partial dependence.
%Contribution
To this end, we take inspiration from classic time series models. The problem of limited feasible autoregressive orders in time series models also occurs in the AR($p$) and ARCH($p$) models. The solution for these models is to introduce moving average terms, resulting in the ARMA$(p,q)$ and GARCH$(p,q)$ models, respectively (see \cite{engle1982autoregressive}, \cite{bollerslev1986generalized}, \cite{bollerslev1994arch}, \cite{francq2019garch} and \cite{hamilton2020time}). By inverting the corresponding linear filter, AR($\infty$) and ARCH($\infty$) representations can be derived. Hence the time series can be understood as having no Markov restriction. This results in the time series models being more persistent in the mean in the case of ARMA and more persistent in the variance in the case of GARCH. The latter can be interpreted as the model being able to capture volatility clusters \citep{bollerslev1994arch}. See also Fig.~\ref{Fig:ARMA} and Fig.~\ref{Fig:GARCH}.
\begin{figure}[t]
\centering
\includegraphics[scale=0.5]{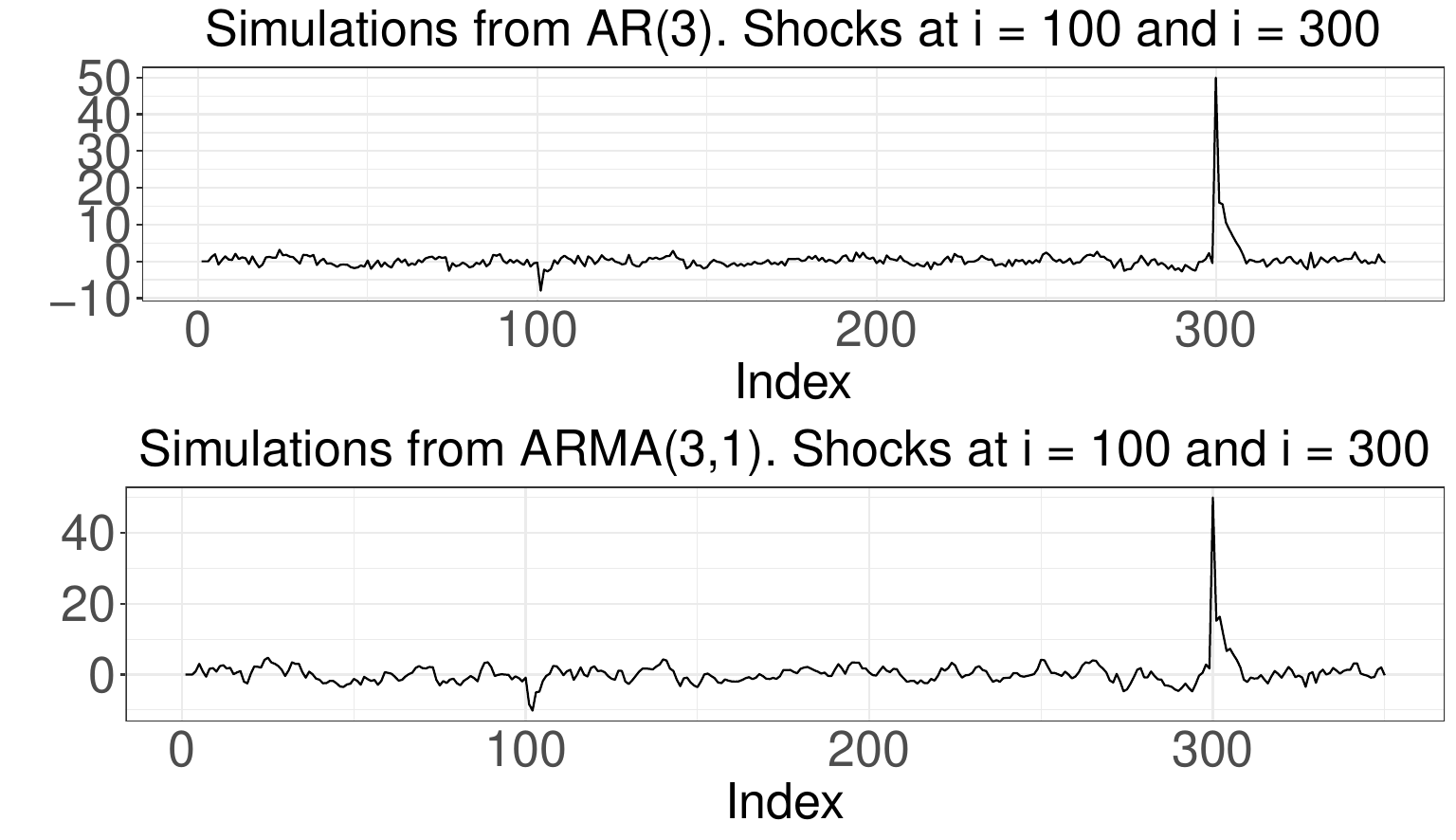}
\caption{Simulations of an AR$(3)$ and a ARMA$(3,1)$ process using the same random number generation. The ARMA process exhibits long-term autoregressive effects in the mean, i.e. smoothness.}
\label{Fig:ARMA}
\end{figure}

\begin{figure}[t]
\centering
\includegraphics[scale=0.5]{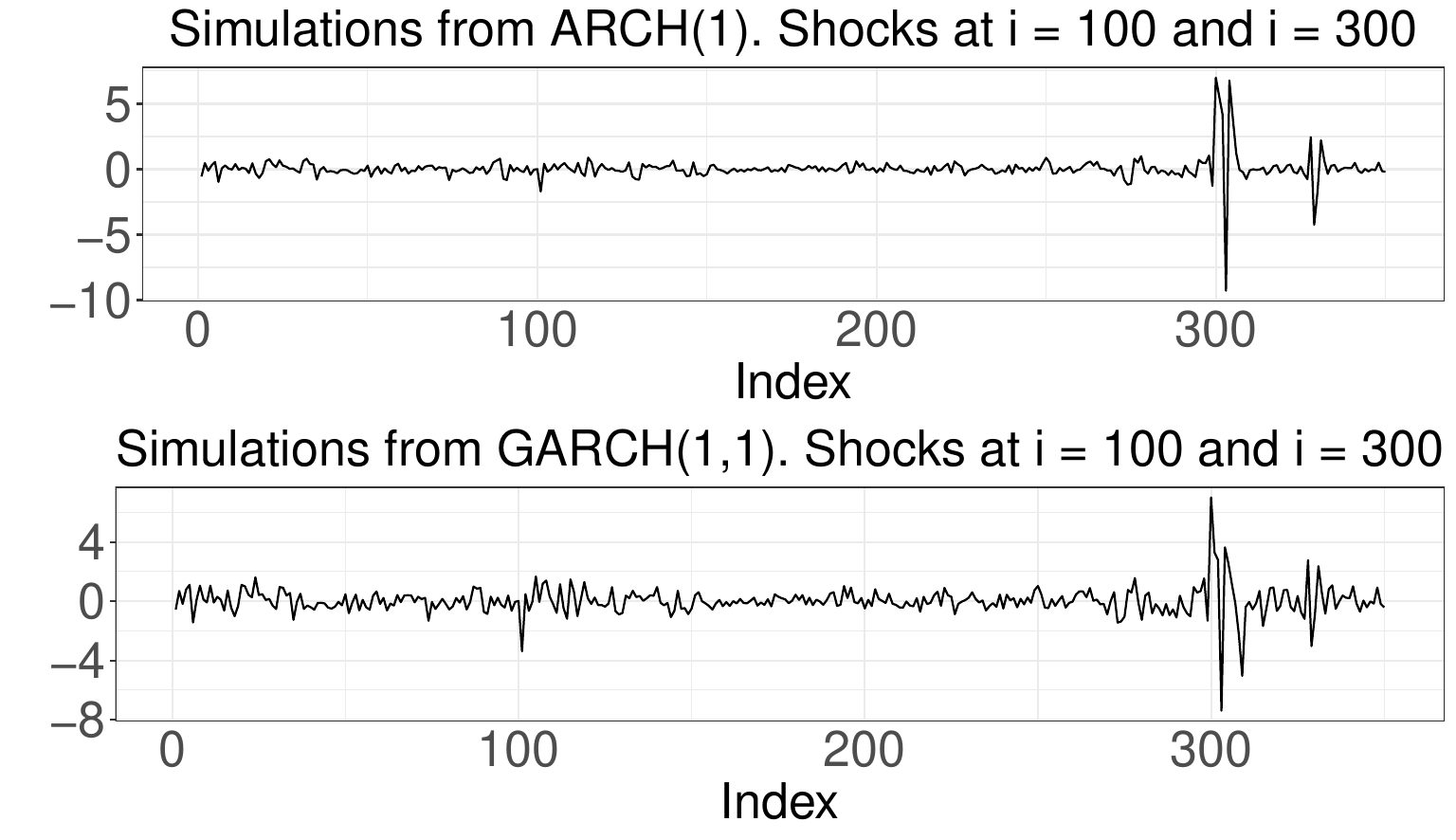}
\caption{Simulations of an ARCH$(1)$ and a GARCH$(1,1)$ process using the same random number generation. The GARCH process exhibits long-term autoregressive effects in the variance, i.e. volatility clustering.}
\label{Fig:GARCH}
\end{figure}
\noindent
Inspired by these approaches, we modify the copula-based time series model by incorporating a moving aggregate (MAG) part in the updating equation. \SP{This allows for long-term autoregressive effects in the whole conditional distribution.} The MAG part is a non-linear function taking the current and the past $q$ innovations as arguments. Formally it is a mapping from $(0,1)^{q+1}$ to $(0,1)$. The functional form of the MAG-function is given by a conditional quantile function corresponding to a copula. We call it 'moving aggregate', because, similar to the moving average, it aggregates the current and past $q$ innovations but it is \SPP{not a} linear aggregation, hence 'aggregate' instead of 'average'. Analogous to the classical ARMA and GARCH models, the MAG-part permits infinite partial dependence. This means that the partial copula \citep[see][]{spanhel2016partial} of $X_t$ and $X_{t-j}$ is in general different from the independence copula for all $j,t \in \mathbb{Z}$. We call the resulting model \textit{MAG-Modified Autoregressive Copula-Based Time Series Model}, or in short: \textit{MAGMAR-Copula}.
%Contributions
In this paper we present the model and show that the classical copula-based time series model and the classical linear ARMA model are nested within the proposed model. We proceed to derive its distributional properties in a D-vine framework (see \cite{aas2009pair, czado2010pair, czado2022vine} and Sect.~\ref{Subsection:Methods:Copulas} for D-vine copulas). AR($\infty$) and MAG($\infty$) representations are derived. Building upon the MAG($\infty$) representation, it is first proven that the \SPP{MAGMAR$(1,1)$-Copula} model, where the AR-copula is the normal copula, is stationary when excluding perfect comonotonicity and perfect countermonotonicity. Second, we establish stationarity and ergodicity under a contraction-on-average property of the time series. \SPP{In general the stationary distribution of the time series is not uniform (as is desired). We introduce an auxiliary adjustment transformation to account for this problem. The transformation guarantees that the unconditional distribution of the time series can always be modeled properly.}
We derive the conditional distribution of $X_t|\mathcal{F}_{t-1}$, where $\mathcal{F}_{t-1}$ denotes the $\sigma$-field of all information up to and including $(t-1)$. Based on conditional independence, the likelihood function is derived. Similar to ARMA and GARCH models, the likelihood has to be calculated in an iterative scheme. The algorithm is presented and asymptotic properties of the estimation are derived. Since asymptotic results for the estimation of Markov copula-based time series models are available in the literature (see \cite{chen2006estimation} or \cite{nagler2022stationary}), we focus on the consistency of the maximum-likelihood estimation of the MAG$(1)$-copula model. \SPP{We proceed to systematically compare our proposed model with the generalizations of copula-based time series by \cite{mcneil2022time} and \cite{joe2014dependence} (chapter 3.14). These models also aim to overcome the Markov restriction/generalize the linear ARMA model.}
\SPP{Last, we test the model in a simulation study and apply it to quarterly US inflation.} We compare the modeling performance of the MAGMAR-Copula model to the model by \cite{mcneil2022time} and to linear modeling (i.e. choosing normal copulas for the temporal dependence structure). We find that our model performs well in this setting and is competitive with the model from \cite{mcneil2022time}. In line with their results, we find that proper modeling of asymmetric temporal tail dependence as well as infinite partial dependence is beneficial. Furthermore, we find that incorporating heavy-tailed temporal dependence (i.e. incorporating the $t$-copula) enhances the performance, indicating not only asymmetric tail-dependence but also heavy-tailed dependencies in US inflation.\\
In the next section, Sect.~\ref{Section:Methods}, the copula-based time series model is introduced and reviewed \SPP{with a focus on D-vines. Furthermore related generalizations of copula-based time series models are reviewed.} In Sect.~\ref{Section:MAGMAR} the Moving Aggregate Modified Autoregressive Copula-Based Time Series Model (MAGMAR-Copula) is introduced \SPP{and theoretical properties are derived.} \SPP{In Sect.~\ref{Section:Simulation_Study}, we explore the model in a simulation study, verifying theoretical findings, extending them and comparing the model with related ones.} The MAGMAR-Copula model is applied to US Inflation in Sect.~\ref{Section:Application}. Sect.~\ref{Section:Conclusion} concludes the work and gives an outlook to future research.
\section{Copula-based Time Series Models}
\label{Section:Methods}
\noindent In this section the notion of copula-based time series models is introduced. First, we give a short introduction to copulas. Then, general ideas of copula-based time series modeling are developed. Last, we discuss the approaches \SPP{to generalize copula-based time series modeling} by \cite{mcneil2022time} and \cite{joe2014dependence}.
%{t = 1}^T
\subsection{Copulas and D-vine Copulas}
\label{Subsection:Methods:Copulas}
\noindent In this section we give a short introduction to copulas and the main ideas to multivariate modeling with copulas. For detailed and comprehensive presentations we refer to standard works: \cite{joe2014dependence} and \cite{nelsen2006introduction}. \\
Consider $p$ continuous random variables $X_1, \hdots, X_p$. Let their joint distribution be given by $F_{X_1,\hdots,X_p}$ and their marginal distributions by $F_{X_1}, \hdots, F_{X_p}$, respectively. The general idea of multivariate copula modeling amounts to decomposing the joint distribution into marginal distributions and a copula,
\begin{align}
F_{X_1,\hdots,X_p}(x_1, \hdots, x_p) = C[F_{X_1}(x_1), \hdots, F_{X_p}(x_p)].
\label{Sklar}
\end{align}
Sklar's theorem states that there exists a copula that fulfills the statement from Eq.~\ref{Sklar}, which is unique if the random variables are continuous \citep{sklar1959fonctions, joe2014dependence}. The copula $C$ is the joint distribution of probability integral transformed (PIT) and hence uniform random variables $F_{X_1}(X_1), \hdots, F_{X_p}(X_p)$ \citep[see][for the PIT]{angus1994probability}, $F_{F_{X_1}(X_1), \hdots, F_{X_p}(X_p)} = C$. The uniformity of $U_i := F_{X_i}(X_i), \ i \in \{1, \hdots, p\}$ allows the copula to be interpreted as the dependence structure of the random variables $X_1,\hdots,X_p$. An important notion for visualization, estimation and for D-vine copulas is the copula density. The copula density of a copula $C$ is given as \citep{joe2014dependence}
\begin{align}
c[u_1, \hdots, u_p] = \frac{\partial^p C[u_1, \hdots, u_p]}{\partial u_1 \hdots \partial u_p}.
\label{CopulaDensity}
\end{align}
Since the copula density is given by the partial derivative of the respective copula, Eq.~\ref{CopulaDensity} is only valid if the copula is differentiable. In applications, the copula and the marginal distributions have to \SPP{be} chosen appropriately. The normal, $t$, Gumbel and independence copulas are described in \ref{App:Subsection:Copulas_in_this_work}.\\
Now we introduce the notion of D-vine copulas. D-vine copulas are a special case of pair-copula constructions. A $p$-variate copula is decomposed into pair-copulas, i.e. bivariate copulas \citep{bedford2001probability, bedford2002vines, aas2009pair, czado2010pair, czado2022vine}. The decomposition is based on the law of total probability and is not unique. The D-vine decomposition amounts to a special ordering of the variables in the decomposition. The decomposition in terms of the $p$-variate copula density is given as \citep{aas2009pair},
\begin{align}
c[u_1, \hdots, u_p] =& \prod_{j = 1}^{p - 1} \prod_{i = 1}^{n -j} c_{i, i+j | i + 1,\hdots,i+j-1}\Big(F(u_i|u_{i+1},\hdots,u_{i+j-1}), \nonumber
\\
& F(u_{i+j}|u_{i+1},\hdots,u_{i+j-1}) \Big),
\label{D-vine}
\end{align}
where $c_{i,i+j| i + 1,\hdots,i+j-1}$ is the conditional copula density of $U_i$ and $U_{i+j}$ given the intermediate variables $U_{i+1},\hdots,U_{i+j-1}$. $F(u_i|u_{i+1},\hdots,u_{i+j-1})$ is the conditional distribution function of $U_i$ given $U_{i+1},\hdots,U_{i+j-1}$ evaluated at $u_i$ and conditioning variables $u_{i+1},\hdots,u_{i+j-1}$. Since the joint distribution of $U_i,U_{i+1},\hdots,U_{i+j-1}$ is a copula, $F(u_i|u_{i+1},\hdots,u_{i+j-1})$ may also be understood as a conditional distribution function, corresponding to a copula. The decomposition in Eq.~\ref{D-vine} is called 'D-vine' for 'drawable vine' because the variables are ordered in a path like manner \citep{czado2022vine}. The vine copula approach allows for more flexibility when modeling multivariate dependencies because each bivariate dependence can be specified separately.
\subsection{Copula-based Time Series with Autoregressive Order of One}
\label{Subsection:Methods:Copula_Based_Time_Series}
\noindent \SPP{Consider a strictly stationary time series $\{X_t\}_{t \in\mathbb{Z}}$. A time series being Markov$(p)$ means that $\mathbb{P}(X_t \leq a| X_{t-1} = x_{t-1}, \hdots, X_1 = x_1) = \mathbb{P}(X_t \leq a | X_{t-1} = x_{t-1}, \hdots, X_{t-p} = x_{t-p}), \ t > p$ \citep{ibe2013markov}. A stationary and univariate Markov$(p)$ time series is characterized by the $(p+1)$-variate joint distribution $F_{X_{t},\hdots, X_{t-p}}$ of consecutive random variables $X_{t}$ to $X_{t-p}$.} The idea is to decompose the joint distribution into a copula and marginal distributions. By Sklar's Theorem there always exists a copula $C := F_{F_X(X_{t}), \hdots, F_X(X_{t-p})}$ such that the joint distribution can be decomposed into marginal distributions $F_{X_t}$ to $F_{X_{t-p}}$ and the copula \citep{sklar1959fonctions, joe2014dependence}. If the random variables $\{X_{t}\}$ are continuous, the decomposition is unique. Due to strict stationarity we can write $F_{X_{t}} = \hdots = F_{X_{t-p}} =: F_X$. Then the decomposition may be written as
\begin{align*}
F_{X_t, \hdots, X_{t-p}}(x_0,\hdots,x_p) = C[F_X(x_0),\hdots, F_X(x_{p})].
\end{align*}
Hence, the joint distribution is specified by a marginal distribution $F_X$ and the copula $C$. Note that the copula is the joint distribution of the probability-integral transformed random variables $U_t := F_X(X_t)$ \SPP{to} $U_{t-p} := F_X(X_{t-p})$. The copula may be parametrized by dependence parameter(s) $\phi$, so \SPP{we write $C = C_{\phi}$}. We introduce this notation to emphasize the analogy to the linear AR model and later on, in Sect.~\ref{Section:MAGMAR} to the linear ARMA model. Then, analogous to the linear AR model, \SPP{the Markov($p$) model can be formulated in \SPP{terms of a updating equation} as}
\begin{align*}
X_t = F^{-1}_X\left(C_{\phi}^{-1}[w_t|F_X(X_{t-1}), \hdots, F_X(X_{t-p})]\right), \ \  \ w_t \stackrel{iid}{\sim} U(0,1).
\end{align*}
\SPP{where $C_{\phi}^{-1}$ is the conditional quantile function corresponding to the copula $C_{\phi}$. Formally $C_{\phi}^{-1}[x|u_1,\hdots,u_p]$ is the solution to $C_{\phi}[x|u_1,\hdots,u_p] = x$.}
We note that this or similar formulations are also used in other works under different names such as 'updating equation' in \cite{mcneil2022time} or in descriptions on how to sample from copulas in \cite{aas2009pair}.
%The calculation of the new realization, $x_t$, amounts to first sampling a realization from a uniform distribution, $w_t$ and then quantile transforming $w_t$ with the conditional quantile function, corresponding to a copula, obtaining \SPP{$U_t := C_{\phi}^{-1}[w_t|F_X(X_{t-1}), \hdots, F_X(X_{t-p})]$.} Last, $U_t$ is again quantile transformed with the marginal distribution's quantile function $F_X^{-1}$. Using this procedure, it is guaranteed that the dependence structure of \SPP{$(X_t,\hdots,X_{t-p})$} is given by $C_{\phi}$ and that the marginal distribution is given by $F_X$.
This model can be regarded as a generalization of the Gaussian AR$(p)$-model. This can be seen by choosing $C_{\phi}$ as the normal copula and the marginal distribution $F_X$ as the standard normal distribution. Although this is a well-known fact in the literature, see e.g. \cite{chen2006estimation} or \cite{smith2010modeling}, we give the derivation for $p=1$ in \ref{App:Subsection:Proof:Related models} for completeness. 
In the following only the pseudo-observation time series, $\{U_t\}_{t = 1}^T, \ U_t := F_X(X_t)$, is considered. The actual time series can be recovered by quantile transforming every pseudo-observation by $F_X^{-1}$. The dynamics of the time series are solely determined by the relation between the pseudo-observations. \SPP{The model equation of the pseudo-observations is given as,}
\begin{align}
U_t = C^{-1}_{\phi}[w_t | U_{t-1}, \hdots, U_{t-p} ], \ \ w_t \stackrel{iid}{\sim} U(0,1).
\label{Eq:AR(p)}
\end{align}
%for the Markov$(p)$ case. For the special case $p = 1$ we recover
%\begin{align}
%U_t = C^{-1}_{\phi}[w_t | U_{t-1} ], \ \ w_t \stackrel{iid}{\sim} U(0,1).
%\label{Eq:AR(1)}
%\end{align}
%
%The non-linearity of the model is two-fold. First, the relation between random variables $U_t$ and $U_{t-1}$ is a non-linear function. Second, the innovations $w_t$ enter in a non-additive way. What at first sight may seem strange is actually necessary. Choosing the model equation this way is the only possibility to ensure that $U_t,U_{t-1},\hdots$ are uniformly distributed and that $C_{\phi}$ is the copula of consecutive variables. The uniformity of the unconditional distribution of all $U_t, \ t \in \mathbb{Z}$ is due to the properties of conditional quantile functions, corresponding to a copula: Given $U_1, U_2 \stackrel{iid}{\sim} U(0,1)$ it follows that $C^{-1}_{\phi}[U_1| U_2] \sim U(0,1)$. This implies that all $U_t$ are distributed identically. From this it follows that the time series is stationary, because the transitions in the updating equation, Eq.~\ref{Eq:AR(1)}, are homogeneous.\\
%
In order to account for possibly different forms of dependency and to be flexible in modeling, D-vine copulas \citep{bedford2001probability, bedford2002vines, aas2009pair, czado2010pair, czado2022vine} can be employed for modeling the temporal dependence as is also done in \cite{nagler2022stationary} and \cite{smith2010modeling}. In this approach the joint density of $(p+1)$ consecutive pseudo-observations is decomposed into a product of bivariate copulas. The decomposition is given by \citep{smith2010modeling, mcneil2022time, nagler2022stationary}
\begin{align}
c_{t, \hdots, t-p}(u_t, \hdots, u_{t-p}) =& \prod_{k = t - p}^{t - 1} \prod_{j = k + 1}^{t} c_{j-k,j|S_{j-k,j}}\Big[F_{j-k|S_{j-k,j}}(u_{j-k}),  \nonumber
\\
& F_{j|S_{j-k|j}}(u_j) \Big],
\label{d-vine_density}
\end{align}
where the shorthand notation $c_{t, \hdots, t-p} = c_{U_t, \hdots, U_{t-p}}$ is used. Furthermore, $S_{j-k, j} = \{j - k + 1, \hdots, j - 1 \}$. As in the cited literature, only the simplified D-vine is considered, dropping the conditioning in the bivariate copula densities, $c_{j-k,j|S_{j-k,j}} = c_{j-k, j}$.
In the case where $c_{j-k,j} = c_{i - k, i} := c_{k} \ \forall j,i \in \{k+1, \hdots, T\}$, the joint density of all random variables \SPP{$\{U_t\}_{t \in\mathbb{Z}}$} admits translational invariance. \cite{nagler2022stationary} show that such processes are stationary. \SPP{Hence we refer to such D-vine copulas as \textit{stationary} D-vine copulas.} They also may be called s-vine ('s' for stationary) \citep{mcneil2022time}. Furthermore, the bivariate copulas, $C_{j-k, j}$ may be parametrized by parameters $\phi_1, \hdots, \phi_p$, where $\phi_1$ governs the dependence of $U_{t-1}$ and $U_t$ $\forall t \in \mathbb{Z}$, $\phi_2$ the dependence of $U_{t}|U_{t-1}$ and $U_{t-2}|U_{t-1}$ and so on. Note that depending on the chosen copula, each $\phi_i, i \in \{1,\hdots,p\}$ may also be a vector containing more than one parameter. \\
It is common to denote conditional distributions of pseudo observations by the $h$-function \citep{aas2009pair, czado2010pair, czado2022vine}. The $h$-functions associated to a copula $C$ are defined as
\begin{align}
h^1(u_1, u_2) = \frac{\partial C[u_1, u_2]}{\partial u_1}, \nonumber
\\
h^2(u_1, u_2) = \frac{\partial C[u_1, u_2]}{\partial u_2}.
\label{hfunctions}
\end{align}
If no superscript is given, then conditioning on the second variable is implied. \SPP{Using this notation,} the conditional distribution \SPP{of e.g.} $F_{U_t|U_{t-1}}$ may be rewritten as $F_{U_t| U_{t-1}}(u_t, u_{t-1}) = C_{\phi}(u_t|u_{t-1}) = h^2_{\phi_1}(u_t, u_{t-1})$.
%The same applies to $U_{t-2}|U_{t-1}$: $F_{U_{t-2}|U_{t-1}}(u_{t-2}, u_{t-1}) = C_{\phi}(u_{t-2}|u_{t-1}) = h^1_{\phi}(u_{t-1}, u_{t-2})$. The $h$-functions can be derived from the associated copula. Their calculation amounts to taking the derivative of copulas. Often the $h$-functions can be derived analytically.
Explicit forms of $h$-functions and their inverses can be found in \cite{smith2010modeling} for the normal, $t$, Clayton, Gumbel and Galambos copula.) An advantage of using the $h$-functions is that all conditional distributions can be decomposed into a \SPP{composition} of $h$-functions, i.e. Rosenblatt (forward/backward) functions \citep{mcneil2022time}. %E.g. the distribution of $(U_t|U_{t-1}, U_{t-2})$ which is needed for the partial distribution of $U_t$ and $U_{t-3}$, can be calculated as
%\citep{aas2009pair, joe2014dependence}
%\begin{align*}
%F_{U_t|U_{t-1}, U_{t-2}}(u_t|u_{t-1},u_{t-2}) = h_{\phi_2}^2\big(h^2_{\phi_1}(u_t, u_{t-1}), h^1_{\phi_1}(u_{t-1}, u_{t-2})\big).
%\end{align*}
%The notion of decomposing a conditional distribution into a function of $h$-functions is named Rosenblatt (forward/backward) function \citep{mcneil2022time}.
The decomposition of $U_t|U_{t-1},\hdots,U_{t-p}$ can be called Rosenblatt backward and the decomposition of $U_{t-p}|U_{t-p+1},\hdots,U_{t}$ Rosenblatt forward function (inspired by the Rosenblatt transform from \cite{rosenblatt1952remarks}). For $p = 1$ the Rosenblatt functions are just the corresponding $h$-functions, $R_1^1 = h_1^1$, $R_1^2 = h_1^2$. \cite{mcneil2022time} define the Rosenblatt forward and backward functions for $p > 1$, $R_p:(0,1)^p\times(0,1) \rightarrow (0,1)$, recursively as follows:
\begin{align}
R_p^1(\mathbf{u}, x) =& h_p^1 \left( R_{p-1}^2(\mathbf{u}_{-1}, u_1), R_{p-1}^1(\mathbf{u}_{-1}, x) \right), \nonumber
\\
R_p^2(\mathbf{u},x) =& h_p^2 \left(R_{p-1}^2(\mathbf{u}_{-p}, x), R_{p-1}^1(\mathbf{u}_{-p}, u_p) \right),
\label{Rosenblatt}
\end{align}
where $\mathbf{u}_{-j} = (u_1,\hdots,u_{j-1},u_{j+1},\hdots,u_k)$. The second argument of the Rosenblatt functions, $x$, is the variable that the Rosenblatt functions are evaluated at. The conditioning variables are contained in the vector $\mathbf{u}$. Using the newly introduced notation, the D-vine copula density (Eq. \ref{d-vine_density}) can be rewritten as \citep{mcneil2022time},
\begin{align*}
c_{t, \hdots, t-p}(u_t, \hdots, u_{t-p}) = \prod_{k = t - p}^{t - 1} \prod_{j = k + 1}^{t} c_{j-k,j}\Big[& R^1_{k - 1}((u_{j - k + 1},\hdots, u_{j - 1}),u_{j}),
\\
& R^2_{k - 1}((u_{j - k + 1},\hdots, u_{j - 1}), u_{j - k}) \Big].
\end{align*}
In order to rewrite the time series updating equation (Eq. \ref{Eq:AR(p)}) in the D-vine framework, the inverse Rosenblatt transforms can be utilized:
\begin{align}
U_t = R^{-1}_p((U_{t-p},\hdots,U_{t-1}), w_t), \quad w_t \stackrel{iid}{\sim} U(0,1).
\label{Eq:UpdatingEqRosenblatt}
\end{align}
\SPP{Here $R_p^{-1}$ is the inverse Rosenblatt function corresponding to the $(p+1)$-dimensional copula $C_{\phi}$ from Eq.~\ref{Eq:AR(p)}. The inversion is w.r.t. the last argument. Formally, $R_p^{-1}(\mathbf{u},x)$ is the solution to $R_p(\mathbf{u},x) = x$.} In this notation, the Markov($p$) restriction becomes obvious.
%The model can also be understood as a generalized AR($p$)-model. In this case, however, the function connecting the past variables, $U_{t-1}, \hdots, U_{t-p}$, innovations $w_t$ and the current observations $U_t$ is not necessarily linear. For the case where $p=1$, the updating equation becomes
%\begin{align}
%U_t = h_{\phi}^{-1}(w_t, U_{t-1}),  \quad w_t \stackrel{iid}{\sim} U(0,1).
%\label{Eq:UpdatingEqhfunction}
%\end{align}
%Here $h_{\phi}$ is the $h$-function associated to the 2-dimensional copula $C_{\phi}$ from Eq.~\ref{Eq:AR(1)}.
Next, we are going to discuss the extension of \cite{mcneil2022time} and \cite{joe2014dependence} to allow for infinite partial dependence.
\subsection{Review of the extension by McNeil \& Bladt}
\label{Subsection:Methods:Bladt_McNeil}
\noindent \cite{mcneil2022time} allow for infinite partial dependence in copula-based time series models by letting the autoregressive order, $p$, be infinite and reparametrizing the joint density with a finite amount of parameters.  The parameters are chosen in such a way that they represent a classical ARMA\SPP{($r,s$)} process with finite $r$ and $s$. To model a time series, the following steps are employed. First, the autoregressive and the moving average order, $r$ and $s$, respectively, are determined by minimizing model selection criteria (e.g. AIC). Next, the partial autocorrelations of the chosen ARMA($r,s$) process are parametrized in terms of the autoregressive and the moving average parameters, $\phi_1, \hdots, \phi_r, \theta_1, \hdots, \theta_s$. Using the relation between autocorrelations and Kendall's $\tau$, the Kendall's partial autocorrelations are obtained. \SPP{In a third step, all partial copulas in Eq. \ref{d-vine_density} are set to a chosen family (e.g. Gaussian or Gumbel) and are reparametrized by Kendall's $\tau$ and consecutively by the AR and MA parameters, which can then be estimated by maximum likelihood. \cite{bladt2025semiparametric} generalize this approach and consider a Gaussian D-vine, i.e. all copulas in Eq.~\ref{d-vine_density} are Gaussian, and then replace a finite amount of selected copulas by non-Gaussian copulas. Henceforth their approach will be called 'IPD' for infinite partial dependence. The conceptual difference to our MAGMAR-copula model and to the model by \cite{joe2014dependence} is that their approach is based on classical ARMA models. Furthermore, the partial copulas are modeled and not the AR- and MAG-copulas. We discuss the conceptual differences in more detail below.}
%One disadvantage of their extension is that the Kendall's $\tau$ reparametrization only works for one-parametric copulas, since only for one-parametric copulas, there is a one-to-one relation. This excludes, for example, the t-copula, which is of particular relevance in practice \citep[see e.g.][]{loaiza2018time, pappert2023forecasting}. Reparametrization for copulas with more than one parameter is more cumbersome.
%
\subsection{Review of the extension by Joe}
\label{Subsection:Methods:Joe}
\noindent
Now we want to briefly introduce the model of \cite{joe2014dependence}, which is to the best of our knowledge the earliest copula extension of the ARMA model. The model \SPP{for a time series $\{U_t\}_{t\in\mathbb{Z}}$} with uniform marginals is given by the following updating equations:
\begin{align}
W_s =& h(\varepsilon_s; W_{s-1},\hdots,W_{s-p}), \ \ s \in \mathbb{Z},
\label{Joe_AR}
\\
U_t =& g(\varepsilon_t; W_{t-q}, \varepsilon_{t-1},\hdots,\varepsilon_{t-q+1}), \ \ t \in \mathbb{Z},
\label{Joe_MA}
\end{align}
where $\varepsilon_t \stackrel{iid}{\sim} U(0,1)$ are the common innovations, $h$ is the conditional quantile function of a $(p+1)$-variate copula and $g$ is the conditional quantile function of a $(q+1)$-variate copula. It is required that the copula corresponding to $h$ is translationally invariant and that the $q$-variate marginal distribution of the copula corresponding to $g$ is the independence copula. The idea is inspired by the Poisson ARMA process of \cite{mckenzie1988some}: First create a process corresponding to an AR process, Eq.~\ref{Joe_AR}, and then feed this process to an MA-like filter together with the common innovations, \SPP{as in} Eq.~\ref{Joe_MA}. In our proposed MAGMAR-copula model, we too attempt to mix AR-type and MA-type copula processes. The difference, however, is that we feed the MA-type process (what we call 'MAG') into the AR-type filter, \SPP{similar to classical ARMA-type models}.
%In other words, in the copula extension to ARMA models by \cite{joe2014dependence} the latent process is an AR-type process whereas in the MAGMAR-copula model, the latent process is an MA-like model. 
%
\subsection{Comparison of the Models with the Proposed Model}
\label{Subsection:Methods:Comparison}
\noindent
\SPP{We briefly point out the main differences between the model by \cite{mcneil2022time}, the model by \cite{joe2014dependence} and our proposed model. In \ref{App:Subsection:Comparison} we compare the models in detail w.r.t. flexibility, computational demand and modeling approach. For brevity we give the conclusions of the comparison here.\\
The MAGMAR and common-innovation model are similar in spirit. However, the MAGMAR-model has a more straight-forward modeling approach. In return, the common-innovations approach can be expected to be more tractable due to the stationary distribution naturally being standard-uniform. The additional adjustment transformation needed for the MAGMAR-model hinders theoretical tractability. The computational demand can be expected to be similar. The flexibility of the common-innovations approach is a potential direction for future research, since it is not obvious due to the AR-part being latent.\\
The modeling approach of the MAGMAR and IPD model are similar as both models can rely on e.g. autodependence and scatter plots. The main difference is that the partial copulas of $(U_t, U_{t-k}|\mathcal{F}_{t-1:t-k+1})$ result from the AR and MAG-copula in the MAGMAR model, and are typically non-standard, whereas in the IPD model, the partial copulas can be chosen directly. This results in the MAGMAR model being preferable if the long-term dependencies are non-Gaussian. If the long-term dependencies are Gaussian and only a few partial dependencies seem to be non-linear, the IPD model is more suitable. While both models have moderate computational demand, the IPD models demand is less.}
%
%MAGMAR is best when the temporal dependencies are complicated and diverse. CICO can be as good but the role of the latent AR is complicated and needs future research. IPD is limited mainly by the convenient choice of only considering one partial dependence structure, needs future research. MAGMARs main problem is the need for the adjustment.}

%
\section{MAGMAR-Copulas}
\label{Section:MAGMAR}
\noindent In this section, we introduce and discuss the moving aggregate modification to copula-based time series models. For illustration, the autoregressive and moving aggregate \SPP{will be} set to one. Results for arbitrary orders $p$ and $q$ can be derived in a similar fashion \SP{but are beyond the scope of this paper;} however, some results with arbitrary orders $p$ and $q$ are given in the Appendix.
\subsection{Introduction to the model}
\label{Subsection:MAGMAR:Introduction}
\noindent Our proposed approach enabling the model to overcome the Markov restriction amounts to modifying the autoregressive updating equation, Eq.~\ref{Eq:UpdatingEqRosenblatt}, by a moving aggregate part. The moving aggregate part is, similar to classical, linear ARMA processes, a function of the current and past innovations. It is given as a conditional quantile function, corresponding to a copula, which we parametrize with parameter(s) $\theta$. \SPP{So instead of $w_t$ in the updating equation, Eq.~\ref{Eq:UpdatingEqRosenblatt}, we have $(R^1_{\theta_1, \hdots, \theta_{q}})^{-1}((w_{t-q}, \hdots, w_{t-1}), w_t)$. Where the inversion is again w.r.t. the last element.} We call the model Moving Aggregate Modified Autoregressive Copula-Based Time Series Model (MAGMAR-Copula model). \SPP{Below we formally define the model with arbitrary model orders $p$ and $q$. The exposition that follows focuses on the case $p=q=1$.} Furthermore, we assume that 1) all copulas are differentiable, 2) that they are within the Fr\'echet-Hoeffding bounds \citep[see][Chapter 2.9]{joe2014dependence} and 3) that the inverses of $h$-functions associated to all used copulas exist. \SP{Now we formally introduce the model.}
\begin{definition}[MAGMAR$(p,q)$-Copula Model]
\label{Def:MAGMARpq}
Let $R_{\phi_1,\hdots,\phi_p} : (0,1)^{p+1} \rightarrow (0,1)$ be a Rosenblatt function associated to a \SPP{stationary} D-vine copula $C_{\phi_1,\hdots,\phi_p}$, which is parametrized by parameters $\phi_1,\hdots,\phi_p$. Furthermore, let $R^{-1}_{\theta_1,\hdots,\theta_q} : (0,1)^{q+1} \rightarrow (0,1)$ be an inverse Rosenblatt function associated to a D-vine copula $C_{\theta_1,\hdots,\theta_q}$ \SPP{with independent margins,} which is parametrized by parameters $\theta_1,\hdots,\theta_q$. The time series $\{U_t\}_{t \in \mathbb{Z}}$ following the updating equation
\begin{align}
R^1_{\phi_1, \hdots, \phi_p}((U_{t-p}, \hdots, U_{t-1}), U_t) = (R^1_{\theta_1, \hdots, \theta_{q}})^{-1}((w_{t-q}, \hdots, w_{t-1}), w_t)
\label{Eq:MAGMAR_p_q}
\end{align}
is a MAGMAR$(p,q)$-time series. The corresponding model is the MAGMAR$(p,q)$-Copula model. The copula $C_{\phi_1,\hdots,\phi_p}$ is the AR-Copula. The copula $C_{\theta_1,\hdots,\theta_q}$ is the MAG-Copula. The LHS of Eq.~\ref{Eq:MAGMAR_p_q} is the AR-part and the RHS is the MAG-part. \SPP{The time series following the updating equation
\begin{align}
U_t = (R^1_{\theta_1, \hdots, \theta_{q}})^{-1}((w_{t-q}, \hdots, w_{t-1}), w_t)
\end{align}
is the \textit{MAG$(q)$-copula model}.} Furthermore, we define the \textit{adjusted MAGMAR$(p,q)$-copula time series} as $\tilde{U}_t = \Psi(U_t), \ \forall t\in\mathbb{Z}$, where $\Psi$ is a distribution function on $(0,1)$.
\end{definition}
\noindent
The Rosenblatt functions in Eq.~\ref{Eq:MAGMAR_p_q} ought to be understood as a composition of $h$-functions (as given in Eq.~\ref{Rosenblatt}). The AR-Rosenblatt function, $R^1_{\phi_1, \hdots, \phi_p}$ is composed of $h$-functions $h_{\phi_1},\hdots,h_{\phi_p}$ while the MAG-Rosenblatt function is composed of $h$-functions $h_{\theta_1},\hdots,h_{\theta_q}$. \SPP{The AR-copula is a stationary D-vine copula, cf. Sect.~\ref{Subsection:Methods:Copula_Based_Time_Series} or \cite{mcneil2022time}. The MAG-copula has independent margins. This property is also imposed by \cite{joe2014dependence}.}
The family of the copulas (e.g. normal, $t$, Gumbel, etc.) has to be chosen. Parameters $\phi$ and $\theta$ have to be set/estimated. We note that the MAG-part, that is used here, was also derived in \cite{mcneil2022time} as the 'causal representation' for the classical copula-based time series models. In \cite{joe2014dependence}, it is also given as the '\SPP{$q$}-dependent copula extension'. The new contribution is that we use this representation to construct the novel MAGMAR-Copula model. The reason for the adjustment, using a distribution function $\Psi$, is the following: The unconditional stationary distribution of a MAGMAR\SPP{$(p,q)$-copula} time series, $\{U_t\}$ is in general not uniform if $p,q>0$. The distribution that arises is a distribution on $(0,1)$, and it depends on the AR and MAG parameters $(\phi,\theta)$. The probability integral transformation using $\Psi$ then ensures that $\tilde{U}_t\sim U(0,1)$. Notice that we did not prove that the stationary distribution of the MAGMAR\SPP{-copula} time series actually exists yet. \SPP{The corresponding proposition} is given in Sect.~\ref{Subsection:MAGMAR:Distributional_Properties} for the MAGMAR$(1,1)$-copula model. This is the reason for introducing $\Psi$ as a distribution function on $(0,1)$ and not as the stationary distribution of the process in Def.~\ref{Def:MAGMARpq}. In applications, we estimate $\Psi$ with non-parametric methods in a complementary step. In the following, we mostly consider the unadjusted time series $\{U_t\}$, since the transformation does not affect most results. When appropriate, we also give the results for the adjusted time series. \\
Before giving the relations to other models and deriving distributional properties, we proceed to discuss implications of the model equation and further explain the model for the case where $p=q=1$. \SPP{In this case the updating equation is given as
\begin{align}
U_t = h_{\phi}^{-1}(h_{\theta}^{-1}(w_t, w_{t-1}), U_{t-1}).
\label{Eq:MAGMAR(1,1)}
\end{align}}
First, we note that $C_{\theta}$ is \textbf{not} the copula of $w_t$ and $w_{t-1}$. The innovations $w_t$ and $w_{t-1}$ are independent and their copula is the independence copula, $F_{w_t,w_{t-1}} = \Pi$. The expression $h_{\theta}(w_t, w_{t-1})$ is to be understood as a function of $(q+1)$ random variables $w_t$ and $w_{t-1}$, just as $\varepsilon_t - \theta \varepsilon_{t-1}$ in the classical ARMA model is a function of two random variables $\varepsilon_t$ and $\varepsilon_{t-1}$. Second, the copula $C_{\phi}$ is also \textbf{not} the copula of $U_t$ and $U_{t-1}$, however it is connected to the copula of the two random variables. This is analogous to the classical ARMA(1,1) model, where the partial autocorrelation of $X_{t-1}$ and $X_t$ is not solely determined by the AR(1) parameter but also by the MA(1) parameter. (The influence of the MAG-Copula and the AR-Copula to the dependence of $U_t$ and $U_{t-1}$ becomes clear when we derive the joint conditional distribution of $U_t$ and $U_{t-1}$ given $\mathcal{F}_{t-2}$.) The expression $h_{\phi}(h_{\theta}^{-1}(w_t, w_{t-1}), U_{t-1})$ may also be understood as a function of random variables. 
To verify that the MAGMAR$(1,1)$ model equation actually yields the expected behavior to have no Markov restriction i.e. \SPP{an implicit} infinite autoregressive order, the model equation can be iterated. To this aim, an explicit form for the innovations is necessary. It can be directly obtained from Eq. \ref{Eq:MAGMAR(1,1)} as
\begin{align*}
w_t = h_{\theta}(h_{\phi}(U_t, U_{t-1}), w_{t-1}).
\end{align*}
Iteratively substituting for $w_{t-1}$ in Eq. \ref{Eq:MAGMAR(1,1)} yields the AR($\infty$) representation. Utilizing the composition symbol, which acts as $\bigcirc_{i = 1}^{N} f_i = f_1 \circ f_2 \circ \hdots \circ f_N$, the representation can be written as
\begin{align}
w_t =& h_{\theta}(h_{\phi}(U_t, U_{t-1}), \cdot ) \circ h_{\theta}(h_{\phi}(U_{t - 1}, U_{t - 2}), \cdot ) \circ \hdots, \nonumber
\\ 
=& \bigcirc_{i = 0}^{\infty} h_{\theta}(h_{\phi}(U_{t - i}, U_{t - i - 1}), \cdot ).
\label{AR(Infinity)}
\end{align}
Hence, in general, each past observation, $U_{t-i}, i = 1,2,\hdots$ has an actual influence on $U_t$. This is the desired property. Furthermore, the MAG($\infty$) representation can be derived as
\begin{align}
U_t =& h_{\phi}^{-1}(h_{\theta}^{-1}(w_t, w_{t-1}), \cdot) \circ h_{\phi}^{-1}(h_{\theta}^{-1}(w_{t - 1}, w_{t - 2}), \cdot) \circ \hdots, \nonumber
\\
=& \bigcirc_{i = 0}^{\infty} h_{\phi}^{-1}(h_{\theta}^{-1}(w_{t - i}, w_{t - i - 1}), \cdot). 
\label{MAG(Infinity)}
\end{align}
As noted before, there also exists a MAG($\infty$) representation for the case with no moving aggregate extension. The representation was derived in \cite{mcneil2022time} as 'causal representation'. It is easy to see that for $C_{\theta} = \Pi$, the two representations are equal.
%TODO Change references proposition
This is in line with proposition~\ref{Prop:Relation_to_other_models}, which proposes that for $C_{\theta} = \Pi$, the MAGMAR(1,1)-Copula model is equal to the classical copula-based time series model. The \SPP{corresponding} proposition is given in Sect.~\ref{Subsection:MAGMAR:RelatedModels}. In Sect.~\ref{Subsection:MAGMAR:Distributional_Properties}, it is explicitly shown that the MAG($\infty$) representation of the MAGMAR-Copula time series converges for the case where the AR-copula, $C_{\phi}$, is the normal copula. For other copulas, it is hard to establish the convergence of the MAG$(\infty)$ representation explicitly. \SPP{However, we can use results from \cite{douc2014nonlinear}, involving a contraction-on-average, to establish desired properties.}
\subsection{Relation to other models}
\label{Subsection:MAGMAR:RelatedModels}
\noindent Now we derive the relation of the proposed MAGMAR-Copula model to the classical copula-based time series model and the classical linear ARMA model. The results are given for the MAGMAR$(1,1)$ case but can be naturally extended to the MAGMAR$(p,q)$ case. Furthermore we define the MAG-copula time series and show that it is stationary.
\SPP{
\begin{prop}
\label{Prop:Relation_to_other_models}
Let $\{U_t\}_{t \in \mathbb{Z}}$ be a MAGMAR$(1,1)$-Copula time series with AR-Copula $C_{\phi}$ and the MAG-Copula, $C_{\theta}$. The following holds.
\begin{itemize}
\item[i)] If $C_{\phi}$ and $C_{\theta}$ are chosen as the normal copulas, the quantile transformed time series $\{\Phi^{-1}(U_t) \}_{t \in \mathbb{Z}}$ is the classical Gaussian and linear ARMA$(1,1)$ model with AR-parameter $\phi$, MA-parameter $\frac{\theta}{(1-\theta^2)^{\frac{1}{2}}}$ and independent normally distributed innovations with variance $(1 - \phi^2)(1 - \theta^2)$.
\item[ii)] If the MAG-copula is the independence copula, $C_{\theta} = \Pi$, then the classical copula-based time series model (see Sect. \ref{Subsection:Methods:Copula_Based_Time_Series}) is recovered.
\item[iii)] If the AR-copula is the independence copula, then the MAG$(1)$-Copula time series is recovered.
\item[iv)] If both the AR and MAG-copula are chosen as the independence copula, the time series $\{U_t\}_{t\in\mathbb{Z}}$ is an $iid$ standard uniform time series.
\end{itemize}
\end{prop}
}
\begin{proof}
See Appendix \ref{App:Subsection:Proof:Related models}.
\end{proof}
\noindent
The first \SPP{two statements} show that the classical linear ARMA model as well as the classical copula-based time series model is nested by the MAGMAR-Copula model. Part ii) also justifies describing the copula-based time series model as MAGMAR$(1,0)$-Copula model in further discussions.
Next we define the MAG-Copula model.
\begin{definition}[MAG(1)-Copula model]
Let $\{U_t\}_{t \in \mathbb{Z}}$ be a MAGMAR$(1,1)$-Copula time series. If the AR-copula is set to be the independence copula, $C_{\phi} = \Pi$, the resulting model is the MAG(1)-Copula model.
\end{definition}
\noindent We note again that the MAG$(1)$-Copula model coincides with the $1$-dependent model from Chapter 3.14 of \cite{joe2014dependence}. It is easy to see that that a MAG(1)-Copula time series, $\{V_t\}_{t\in\mathbb{Z}}$ follows the updating equation
\begin{align*}
V_t = h_{\theta}^{-1}(w_t, w_{t-1}), \ w_t \stackrel{iid}{\sim} U(0,1).
\end{align*} 
Now we establish the stationarity of the MAG(1)-Copula model and show that the unconditional marginal distribution is given as standard uniform.
\begin{prop}
\label{Prop:MAG(1)_Distributional}
Let $\{U_t\}_{t \in \mathbb{Z}}$ be a MAG$(1)$ time series. Then $\{U_t\}_{t \in \mathbb{Z}}$ is strictly stationary, ergodic and $U_t \sim U(0,1), \ \forall t \in \mathbb{Z}$.
\end{prop}
\begin{proof}
\SPP{See~\ref{App:Subsection:Proof:Distributional_properties_MAG(1)}}
\end{proof}
\noindent The models distributional properties are important for further work. They play an important role in establishing the stationarity of the MAGMAR$(1,1)$-Copula model with the normal AR-Copula, see Sect.~\ref{Subsubsection:MAGMAR:Distributional:Stationarity}.
%The next proposition treats the trivial case, where both, the AR and the MAG copula are chosen as the independence copula.
%\begin{prop}
%Let $\{U_t\}_{t \in \mathbb{Z}}$ be a MAGMAR$(1,1)$-Copula time series. If the AR and the MAG-copula are set to be the independence copula, $C_{\theta} = C_{\phi} = \Pi$ the time series $\{U_t\}_{t \in \mathbb{Z}}$ is an $iid$ standard uniform time series.
%\end{prop}
%
\subsection{Distributional Properties of the model}
\label{Subsection:MAGMAR:Distributional_Properties}
\noindent In this section, the distributional properties of random variables $\{U_i\}_{i \in \mathbb{Z}}$, following the MAGMAR-Copula model, are derived. In order to keep the findings comprehensible, the ideas are developed for the MAGMAR$(1,1)$ case. They can be generalized to the MAGMAR($p, q$) case. First, we show that the process is stationary for the case where the AR-copula is the normal copula and the parameter fulfills $|\phi| < 1$. Second, we derive the conditional distribution of $U_t$ given past observations. Third, using the law of total probability and conditional independence we derive the joint distribution of $\{U_i\}_{i \in \mathbb{Z}}$.
%
%\subsection*{Marginal distribution}
%\label{Distributional:Marginal}
%
%If $U_1 \sim U(0,1)$ and $w_1 \sim U(0,1)$, each of the random variables $\{U_t\}_{t=1}^T$ following the MAGMAR$(1,1)$ process is, by virtue of the $h$-functions, uniformly distributed. In general it holds for any $h$-function $h_{\nu}$, associated to some copula $C_{\nu}$, that $h_{\nu}^{-1}(V, W) \sim U(0,1)$, given that $V, W \stackrel{iid}{\sim} U(0,1)$. This can be seen by calculating the joint density of $(W, h_{\nu}^{-1}(V, W))$ as $c_{\nu}$ via density transformation. By definition, the marginals of the copula $c_{\nu}$ are uniformly distributed. This reasoning of course only holds if $U_1 \sim U(0,1)$. Otherwise the distribution would change in every time step. The generalization to arbitrary AR and MAG orders is straight-forward. Since the Rosenblatt functions are composed of $h$-functions, the same argument holds.
\subsubsection{Stationarity}
\label{Subsubsection:MAGMAR:Distributional:Stationarity}
\noindent We first show that the process is stationary for the case where the AR-copula is the normal copula. \SP{Then we generalize the result for the case, where the AR-copula fulfills a contraction-on-average property.} In order to establish stationarity for the normal AR-copula case, we consider the MAG representation of the MAGMAR$(1,0)$-normal-Copula time series and show that it does converge almost surely for all $t$. Together with the stationarity of the MAG$(1)$-Copula model, the stationarity of the MAGMAR$(1,1)$-Copula model will follow. \SPP{The result is a special case of the more general result on MAGMAR$(1,0)$ models, where the copula on average induces a contraction in proposition~\ref{Prop:AR-Copula_Contraction}. We still propose the result for the normal AR-copula for illustrational purpose}
\begin{prop}
\label{Prop:Convergence_normal_MAGMAR(1,0)}
Let $\{U_t\}$ be a MAGMAR$(1,0)$-Copula time series where the AR-Copula is the normal copula,  $C_{\phi} = C_{\phi}^{\text{normal}}$, with parameter $|\phi| < 1$. Then the MAG$(\infty)$ representation,
\begin{align*}
A_{t,n} := \bigcirc_{i = 0}^n h_{\phi}^{-1}(w_{t-i}, \circ)
\end{align*}
does converge almost surely as $n \rightarrow \infty$ and $\lim_{n \rightarrow \infty} A_{t,n} = U_t, \ U_t \sim U(0,1)$ almost surely.
\end{prop}
\begin{proof}
See~\ref{App:Subsection:Proof:Convergence_MAGMAR(1,0)}
\end{proof}
\noindent Notice that the \SPP{almost sure convergence also} holds when instead of an $iid$ sequence $w_t$, a dependent, but stationary, sequence is plugged in. This is in line with Theorem 36.4 from \cite{billingsley1995probability}. Furthermore, from Lemma~\ref{Prop:Convergence_normal_MAGMAR(1,0)} it follows that the MAGMAR$(1,0)$-Copula time series is strictly stationary. We want to note, that the MAG($\infty$)-representation of the MAGMAR$(1,0)$-Copula time series is analogous to the MA($\infty$) representation in the classical linear ARMA model \citep[see][]{hamilton2020time} and is equivalent to the representation derived in \cite{mcneil2022time}. In principle the MA($\infty$) representation of the classical linear ARMA model can also be written as a infinite composition. However, there the composed functions are linear and, hence, also their composition is linear. 
\noindent We noted earlier that the $a.s.$ convergence in proposition~\ref{Prop:Convergence_normal_MAGMAR(1,0)} also holds if the innovation time series $\{w_t\}_{t \in \mathbb{Z}}$ is not $iid$ but merely stationary. In the next proposition we exploit this fact to show that the MAGMAR$(1,1)$-Copula time series with the AR-Copula being the normal copula is stationay.
\begin{prop}
\label{Prop:Stationarity_MAGMAR(1,1)_normal_AR}
Let $\{U_t\}_{t \in \mathbb{Z}}$ be a MAGMAR$(1,1)$-Copula time series where the AR-Copula is the normal copula,  $C_{\phi} = C_{\phi}^{\text{normal}}$, with parameter $|\phi| < 1$. Then $\{U_t\}_{t \in \mathbb{Z}}$ is strictly stationary.
\end{prop}
\begin{proof}
See~\ref{App:Subsection:Proofs:Stationarity_MAGMAR(1,1)_normal_AR}
\end{proof}
\noindent
%Note that even for the simple case where AR and MAG copula are the normal copula, the unconditional stationary distribution of $\{U_t\}_{t \in \mathbb{Z}}$ is not uniform. This is the reason for introducing the adjusted model in Def.~\ref{Def:MAGMARpq}.
Now we generalize the result from proposition~\ref{Prop:Stationarity_MAGMAR(1,1)_normal_AR} for general bivariate copulas that fulfill a contraction-on-average assumption.
\SP{
\begin{prop}[AR-Copula contraction]
\label{Prop:AR-Copula_Contraction}
Let $\{U_t\}_{t\in\mathbb{Z}}$ be a MAGMAR$(1,1)$-Copula time series. Denote the AR copula by $C_{\alpha}$ and the associated inverse $h$-function as $h_{\alpha}^{-1}(\,\cdot, \, \cdot \,)$. If there is a bijective mapping $\kappa: \mathbb{R} \rightarrow (0,1)$ such that $f_z:(0,1)\rightarrow\mathbb{R}$ with $f_z(x) = \kappa^{-1}[ h_{\alpha}^{-1}(z, \kappa(x))]$ fulfills assumptions A4.37 and A4.38 from \cite{douc2014nonlinear}, then $\{U_t\}$ is stationary and ergodic. 
\end{prop}
}
\begin{proof}
See~\ref{App:Subsection:Proofs:AR_Copula_Contraction}
\end{proof}
\noindent
The conditions have to be checked for each AR-copula individually. In particular, a suitable mapping $\kappa$ has to be found. In \ref{App:Subsection:ContractionCopulas} we check the conditions for the normal, $t$ and Clayton copula. \SPP{The generalization of the result for $q>1$ is straight-forward, since the resulting MAG$(q)$ process is still stationary and ergodic. The generalization for $p>1$ is not straight-forward. The typical approach would be to extend the updating equation to a $p$-dimensional updating equation and to show that the corresponding mapping induces a contraction. This approach is not straight-forward for copula-based models because the inverse Rosenblatt transformation is hard to tract and it is not clear how to establish a contraction. This question will be the subject of future research.} Note that the proposition now enables us to consider the adjustment transformation $\Psi$ from Def.~\ref{Def:MAGMARpq} as the unconditional stationary distribution function of the process, since we now proved such a distribution exists \SPP{under suitable conditions.} Now we proceed to derive the conditional distributions of the model.
\subsubsection{Conditional distribution}
\label{Subsubsection:MAGMAR:Distributional:Conditional_Distribution}
\noindent The conditional distribution of $U_t|\mathcal{F}_{t-1}$, where $\mathcal{F}_{s} := \sigma(\{X_{r}\}_{r \leq s})$, can be derived by using the density transformation theorem \citep[see e.g.][Sect. 4.7]{devore2012modern}. We can interpret the updating equation of the MAGMAR$(1,1)$-Copula model conditional on \SPP{$\mathbf{U}_{t-1} := (U_{t - i})_{i=1}^{\infty} = (U_t, U_{t-1},\hdots)$} as follows: The random variable $w_t \sim U(0,1)$ is transformed by a function $g_{\mathbf{U}_{t-1}}:(0,1) \rightarrow (0,1)$ to the random variable $U_t$, \SPP{i.e. $U_t|\mathcal{F}_{t-1} = g_{\mathbf{U}_{t-1}}(w_t)|\mathcal{F}_{t-1}$. The other random variables, $\mathbf{U}_{t-1} = (U_{t - i})_{i=1}^{\infty}$}, are regarded as deterministic due to conditioning. We give an explicit form for $g_{\mathbf{U}_{t-1}}$ in the proof of the following proposition. 
\begin{prop}
\label{Prop:Conditional_Density_MAGMAR(1,1)}
Let $\{U_t\}_{t \in \mathbb{Z}}$ be a MAGMAR$(1,1)$-Copula time series. Furthermore, let the AR-Copula, $C_{\phi}$, and the MAG-Copula, $C_{\theta}$, be differentiable. Then the conditional density of $U_t|\mathcal{F}_{t-1}$ is given as
\begin{align}
f_{U_t|\mathcal{F}_{t-1}}(x) = c_{\theta} \left( h_{\phi}(x, U_{t-1}), w_{t-1}(\{U_{t - i}\}_{i=1}^{\infty}) \right) c_{\phi}(x, U_{t-1}).
\label{MAGMAR11ConditionalDensity}
\end{align}
where $c_{\theta}$ is the copula density associated to the MAG-Copula and $c_{\phi}$ is the density associated to the AR-Copula. The expression $w_{t-1}((U_{t - i})_{i=1}^{\infty})$ is to be understood as a function of conditioning variables $(U_{t - i})_{i=1}^{\infty}$, which are regarded as constant in the conditional density. For the adjusted time series, $\{\tilde{U}_t\}_{t\in\mathbb{Z}} = \{\Psi(U_t)\}_{t\in\mathbb{Z}}$, the conditional density is given as
\begin{align*}
f_{\tilde{U}_t|\mathcal{F}_{t-1}}(x) =& c_{\theta} \left( h_{\phi}(\Psi^{-1}(x), \Psi^{-1}(\tilde{U}_{t-1})), w_{t-1}(\{\tilde{U}_{t - i}\}_{i=1}^{\infty}) \right) c_{\phi}(\Psi^{-1}(x), \Psi^{-1}(\tilde{U}_{t-1}))
\\
&\times \frac{1}{\psi(\Psi^{-1}(x))},
\end{align*}
where $\psi$ is the density corresponding to the distribution function $\Psi$.
\end{prop}
\begin{proof}
See~\ref{App:Subsection:Proof:ConditionalDistributionMAGMAR11}
\end{proof}
\noindent It is evident from the conditional density, Eq.~\ref{MAGMAR11ConditionalDensity}, that $U_t | \mathcal{F}_{t-1} = U_t | (U_{t - i})_{i=1}^{\infty} = U_t | (U_{t-1}, w_{t-1})$. This means that conditioning on all past random variables $(U_{t - i})_{i=1}^{\infty}$ is equivalent to conditioning on $(U_{t-1}, w_{t-1})$. This is the same property as in the classical linear ARMA model.
Note that for the special case where $c_{\theta}$ is the independence copula density, $c_{\theta} = 1_{(0,1)^2}$, the conditional density of the copula-based time series model with autoregressive order of one is recovered. In the MAGMAR$(p, q)$ case, the conditional density can be derived in a similar fashion. The random variable $w_t$ is transformed to the random variable $U_t$ by an appropriate transformation function $g_{\mathbf{U}_{t-1}}$. The resulting density is a rather complicated expression which is given in \ref{App:Subsection:Conditional_Density_MAGMAR(p,q)}. Furthermore we note that multivariate conditional densities such as $f_{U_t,U_{t-1}|\mathcal{F}_{t-2}}$ can be derived in a similar way. The derivation of joint conditional densities for the MAGMAR$(1,1)$ model can be found in \ref{App:Subsection:JointConditionalDistribution}.
\subsubsection{Joint distribution and Likelihood}
\label{Subsubsection:MAGMAR:Distributional:Joint_distribution_Likelihood}
\noindent To derive the joint distribution of $t$ consecutive random variables from the time series $\{U_t\}_{t = 1}^T$ following the MAGMAR-Copula model, conditional independence and the law of total probability \citep[][B.3.5]{schervish2012theory} is utilized. The joint distribution of the $t$ random variables  $(U_t,\hdots,U_1)$ is given as
\begin{align}
f_{U_t, \hdots, U_1}(u_t, \hdots, u_1) 
= f_{U_1}(u_1) \prod_{i = 0}^{t-2} f_{U_{t-i}|\mathcal{F}_{t-i-1}}(u_{t-i}|\mathcal{F}_{t-i-1}).
\label{joint_density}
\end{align}
For the MAGMAR(1,1)-Copula model, the conditional distribution from Eq.~\ref{MAGMAR11ConditionalDensity} can \SPP{be} plugged in for $f_{U_{t-i}|\mathcal{F}_{t-i-1}}$. For the MAGMAR$(p,q)$-model the conditional densities from \ref{App:Subsection:Conditional_Density_MAGMAR(p,q)}, Eq.~\ref{ConditionalDensityMAGMARpq} need to be used.
In spite of the factorization, the infinite partial dependence is carried by the $\{w_i\}_{i = 1}^T$. For maximum likelihood estimation purposes, the initial density, $f_{U_1}$, may be dismissed, hence resulting in pseudo-MLE. It is easy to verify (via integration by substitution) that $\int f_{U_t,\hdots,U_1}(u_t,\hdots,u_1) d u_t = f_{U_{t-1},\hdots,U_1}$, which is to be expected of a joint density. Using the joint density derived in this section, maximum likelihood estimation can be performed. This is discussed in the next section.
\subsection{Estimation \& Modeling}
\label{subsection:MAGMAR:Estimation_Modeling}
\noindent In this section we describe how a univariate, strictly stationary time series may be modeled with the MAGMAR-Copula model. First, we describe the (pseudo-)maximum likelihood estimation given fixed AR and MAG orders and fixed bivariate copulas. Second, we describe challenges in \SPP{modeling and} model selection.\\
Given a fixed AR order $p$, fixed MA order $q$, fixed bivariate copula specifications $C_{\phi_1},\hdots, C_{\phi_p}$ and $C_{\theta_1},\hdots,C_{\theta_q}$ for the respective D-vines and observations $(u_T,\hdots,u_1)$, the estimation can be performed by (pseudo)-maximum likelihood estimation. The likelihood can be derived directly from the joint density (Eq.~\ref{joint_density}). The pseudo-log-likelihood of the MAGMAR$(1,1)$-Copula model of $T$ observations is given by
\begin{align}
\log L(\phi, \theta) =& \sum_{i = 0}^{T - 1} \Bigg( \log  c_{\theta} [ h_{\phi}(u_{T-i}, u_{T-i-1}), w_{T-i-1})] \nonumber
\\
& + \log c_{\phi}[u_{T-i}, u_{T-i-1}] \Bigg).
\label{Likelihood_MAGMAR(1,1)}
\end{align}
The likelihood of the MAGMAR$(p,q)$-Copula model is given in a completely analogous way using the Eq.~\ref{ConditionalDensityMAGMARpq} from \ref{App:Subsection:Conditional_Density_MAGMAR(p,q)} for the conditional densities. \SPP{We note that this likelihood is only valid if the true data-generating process is a MAGMAR$(p,q)$ process.} When marginal distributions and the adjustment transformation $\Psi$ are included, an additional term emerges \SPP{in accordance with proposition~\ref{Prop:Conditional_Density_MAGMAR(1,1)}.} If the MAG-copula is set to the independence copula, the likelihood for the classical copula-based time series emerges. \cite{chen2006estimation} establish desirable asymptotic properties (consistency and asymptotic normality) of the estimators for the case where the marginal distribution is estimated non-parametrically. More generally \cite{nagler2022stationary} establish these properties for multivariate stationary time series and arbitrary lags. The peculiarity of the likelihood in Eq.~\ref{Likelihood_MAGMAR(1,1)} is that it includes the term $\log c_{\theta}[h_{\phi}(u_{T-i}, u_{T-i-1}), w_{T-i-1}]$, originating from the MAG-copula.
The following proposition establishes the consistency of ML-estimators for the MAG$(1)$ model. The extension for the MAG$(q)$ model follows naturally.
\begin{prop}[Consistency of MLE of MAG$(1)$ copula process]
\label{Prop:MAG(1)_Consistency}
Let $\{U_t\}_{t \in \mathbb{Z}}$ be a MAG$(1)$-Copula time series with MAG-copula $C_{\theta_0}$. Let the parameter space $\Theta$ be compact and the MAG-Copula continuously differentiable w.r.t. parameters $\theta \in \Theta$ and continuous in its arguments, $(u,v)\in(0,1)^2$. Furthermore, let the time series \SPP{of residuals,} $\{\hat{w}_t\}_{t\in\mathbb{Z}}$ given as $\hat{w}_t = h_{\theta}(U_t,\hat{w}_{t-1})$, fulfill Assumptions A4.36-A4.38 from \cite{douc2014nonlinear}. Let $l_t(\theta) = - \log c_{\theta}[U_{t}, \hat{w}_{t-1}]$, then the \SP{(pseudo)} maximum-likelihood estimator
\begin{align*}
\hat{\theta} = \argmin_{\theta \in \Theta} \sum_{t = 2}^T l_t(\theta)
\end{align*}
is a consistent estimator for the unique $\theta_0 = \argmin_{\theta \in \Theta}  \mathbb{E}(l_t(\theta))$.
\end{prop}
\begin{proof}
See~\ref{App:Subsection:Proof:MAG(1)_consistency}
\end{proof}
\noindent
We note that the crucial step in establishing consistency is the stationarity and ergodicity of the time series of residuals, $\{\hat{w}_t\}$. To this end, theorem 4.40 from \cite{douc2014nonlinear} is used. The theorem is based on a contraction-on-average, which has to be established by verifying Assumptions A3.36-A3.38 for each copula individually. The calculations are similar to the ones for establishing stationarity and ergodicity of the MAGMAR$(1,1)$ time series model. \SPP{The major difference is that the $U_t$, which act as innovations are not independent and that the updating equation involves the conditional CDF associated to the MAG-copula and not the conditional quantile function. A natural extension of proposition~\ref{Prop:MAG(1)_Consistency} would be asymptotic normality. However, to establish such result a uniform WLLN for the Hessian matrix is needed, cf. condition M.5(i) of Theorem 4.4 by \cite{wooldridge1994estimation}. This is beyond the scope of this paper. Nevertheless, the simulation study in Sect.~\ref{Subsection:Simulation_Study:Asymptotics} indicates that asymptotic normality holds for the MAGMAR$(1,1)$ model.}\\
The likelihood has to be calculated using an iterative scheme \SPP{and initial values have to be set.} The algorithm for calculating the likelihood of a MAGMAR$(p,q)$-Copula time series \SPP{with} given AR \& MAG orders, $p$ \& $q$, fixed bivariate copula specifications, $C_{\phi_1},\hdots, C_{\phi_p}$ and $C_{\theta_1},\hdots,C_{\theta_q}$ for the respective D-vines, fixed parameters $(\phi_1,\hdots, \phi_p, \theta_1, \hdots, \theta_q)$ and observations $(u_T,\hdots,u_1)$ is \SPP{in spirit the same as for ARMA or GARCH models. For completion we provide the detailed steps in~\ref{App:Subsection:Additional:Likelihood_Iteration}.}\\
Now we turn our attention to challenges \SPP{in general modeling. While the MLE is well-behaved, if it is assumed that the observed data is generated from a MAGMAR$(p,q)$ time series model, there are some subtleties in modeling real data. When modeling a time series with copula-based approaches, the first step is to estimate the unconditional stationary distribution of the time series and transform the data to pseudo-observations, which are distributed standard-uniformly if the unconditional stationary distribution is correctly specified. The problem with this approach in the MAGMAR-context is that the unconditional stationary distribution of the regular MAGMAR model is not standard-uniform. This is why the adjusted MAGMAR$(p,q)$ model is introduced in Def.~\ref{Def:MAGMARpq}. In principle it is possible to jointly estimate the adjustment transformation $\Psi$ and the MAGMAR model, however, this approach is computationally very demanding. This is why we suggest the following approach: First we fit a regular MAGMAR-model on the pseudo-observation data. This will be misspecified. Nevertheless, we then use the estimated Model to obtain the implied non-standard-uniform stationary distribution and quantile-transform the pseudo-observation data w.r.t. this distribution. Then the estimation is repeated on this transformed data. This procedure can be repeated. In our simulation study we show that two iterations yield satisfactory results on a MAGMAR$(1,1)$ model.}
\\
\SPP{Now regarding model selection,} in contrast to linear models (i.e. ARMA), model selection in the MAGMAR-Copula context does not only amount to selecting appropriate model orders, i.e. AR \& MAG orders, but also appropriate copulas. In popular vine copula software such as the \texttt{VineCopula} package in \texttt{R} \citep{schepsmeier2015package}, 38 bivariate copula specifications are available. Hence estimating all possible models is computationally expensive. In order to be able to perform sensible model selection, we suggest not solely relying on information criteria but also doing pre-selection based on descriptive statistics i.e. pair plots, acf, pacf or kpacf \citep[see][for the kpacf]{mcneil2022time} and tail dependencies. \SPP{Similar to the (p)acf for classical linear ARMA models, (partial) auto-dependence measures can give information about sensible lag orders for the MAGMAR-copula model. If the partial auto-dependence measures are cut-off after some lag order, a MAG part is not necessary, because a cutoff partial auto-dependence is equivalent to the process being Markov, assuming that the relations are monotone. We formally prove this in \ref{App:Subsection:Autodependence_Properties}.}
%On the other hand, if the auto dependence measures are cut-off, the process is $q$-dependent and a MAG-term is sufficient. This property follows immediately from the definition of $q$-dependent processes.
When the partial auto-dependence measures are not cut-off, a long auto-dependence is indicated hence requiring inclusion of a MAG part in the model. The tail dependence measures can give information about sensible bivariate copula choices. No tail dependence indicates linear relations, while symmetric tail dependence indicates heavy-tailed dependence. Asymmetric tail dependence naturally indicates asymmetric copulas. \\
For further use, we introduce new notation. First, we abbreviate the t, normal, Gumbel, \SPP{Joe} and independence copulas by 't', 'n', 'g', 'j' and 'i', respectively. By the following example, we explain the new notation: We let MAGMAR($3,2$)-(n,g,t)-(i,t) denote the MAGMAR-Copula model with 3 AR-copulas and 2 MAG-copulas. The first AR-copula, governing the dependence of $U_t$ and $U_{t-1}$, is the normal copula, indicated by 'n'. The second AR-copula is the Gumbel copula indicated by 'g' and the third AR-copula is the t-copula, indicated by 't'. The first MAG-copula is the independence copula, indicated by 'i' and the second MAG-copula is the t-copula, indicated by 't'.
In the next section we model the US inflation with the MAGMAR-Copula model using the likelihood derived in this section.
\section{Simulation Study}
\label{Section:Simulation_Study}
\noindent
\SPP{In order to test the proposed model, we conduct the following simulation studies. First, we explore asymptotic properties of the estimators. To this end we re-run the MLE and use QQ-plots to assess the validity of asymptotic normality and consistency. Second, we exemplify the modeling approach on a toy example and last we compare the approach with other copula-based approaches.}
\subsection{Asymptotics}
\label{Subsection:Simulation_Study:Asymptotics}
\noindent
\SPP{We are interested in consistency and asymptotic normality of the maximum likelihood estimators in the MAGMAR-model. To this end, we simulate $m=200$ realizations of a MAGMAR$(1,1)$-(g)-(n) process with length $n = 1000$. The parameters are chosen as $\phi = 1.4$ for the Gumbel AR-copula and $\theta = 0.3$ for the Gaussian MAG-copula.
Based on each of the $m$ realizations of the process, we estimate the parameters of the above model by MLE, leaving us with an ensemble of $m=200$ realizations of the AR and MAG parameter estimators $\hat{\phi}$ and $\hat{\theta}$. Based on these estimators consistency and asymptotic normality may be inspected. To this end the QQ-plots of both CLT statistics of the estimators are inspected. The CLT statistics are calculated as $Z^{(\phi)}_{i,n} = \frac{\hat{\phi}_i - \phi_0}{\hat{\text{sd}}(\hat{\phi})}$ (and analog for $\theta$). The resulting QQ-plots are given in Fig. \ref{Fig:QQplot_MAGMAR(1,1)}. We observe that 1) the estimators in fact seem to be fluctuating around the true underlying value and 2) the fluctuations scale with $\sqrt{n}$ and approach normality for large-sample sizes, as is expected. As a side effect of repeating the estimation $m=200$ times for the MAGMAR$(1,1)$ model, we can also estimate the average computational demand for estimating the parameters of such model. The estimated average run time is roughly $1min$ per estimation.}\\
\SPP{Next we investigate whether the iterative fitting approach described above leads to asymptotic normality. To this end we simulate from an adjusted MAGMAR-(g)-(t) model with $\phi = 1.4$ and $(\rho, \nu) = (0.3, 4.2)$ and quantile-transform the simulated time series w.r.t. the standard normal quantile function. This is a more realistic modeling scenario than the one above. To estimate the adjusted MAGMAR-copula model, we first transform the time series $\{Y_t\}$ to pseudo-observations $\{U_t\}$ by applying the PIT w.r.t. the empirical CDF.
Then we perform two iterations of estimating the copula parameters and the stationary distribution, as described in Sect.~\ref{Subsubsection:MAGMAR:Distributional:Joint_distribution_Likelihood}.
%Then, in the first iteration, we estimate the copula parameters and subsequently the implied marginal distribution $\Psi^{(1)}$. Now in the second iteration, the pseudo-observations are quantile-transformed w.r.t. $\Psi^{(1)}$, the MAGMAR-copula models parameters are estimated on the transformed observations again and the implied marginal distribution $\Psi^{(2)}$ is extracted. 
This procedure is repeated $m=200$ times. We observe that the first-iteration estimators are in fact biased. The average of e.g. $\{\hat{\phi}_i\}_{i = 1}^m$ is $\bar{\phi} = 1.56$ with the true value being $\phi = 1.4$. This is expected because the marginal distribution of the pseudo-observations and the one of the MAGMAR model differ. The estimators of the second iteration, however, are well-behaved with average estimates $(\bar{\phi}, \bar{\rho}, \bar{\nu}) = (1.41, 0.29, 4.50)$. In Fig.~\ref{App:Fig:QQplot:MAG(1)_par1} in \ref{App:Section:Additional_Plots} we also provide the QQ-plots of the respective CLT-statistics, which indicate asymptotic normality for the estimators $\hat{\phi}$ and $\hat{\rho}$. Asymptotic normality for $\hat{\nu}$ does not seem to hold.}
\begin{figure}
\centering
\includegraphics[scale=0.3]{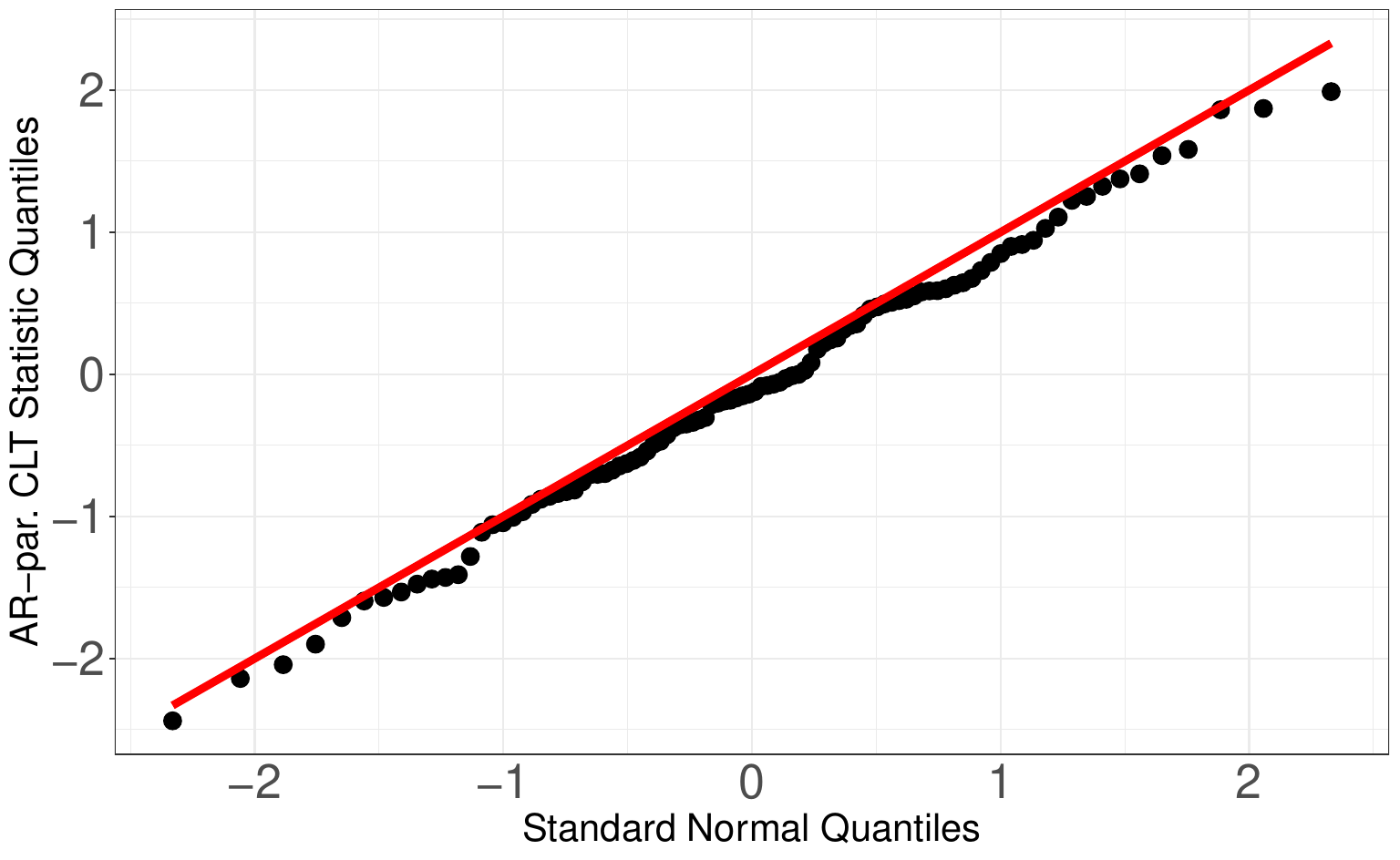}
\includegraphics[scale=0.3]{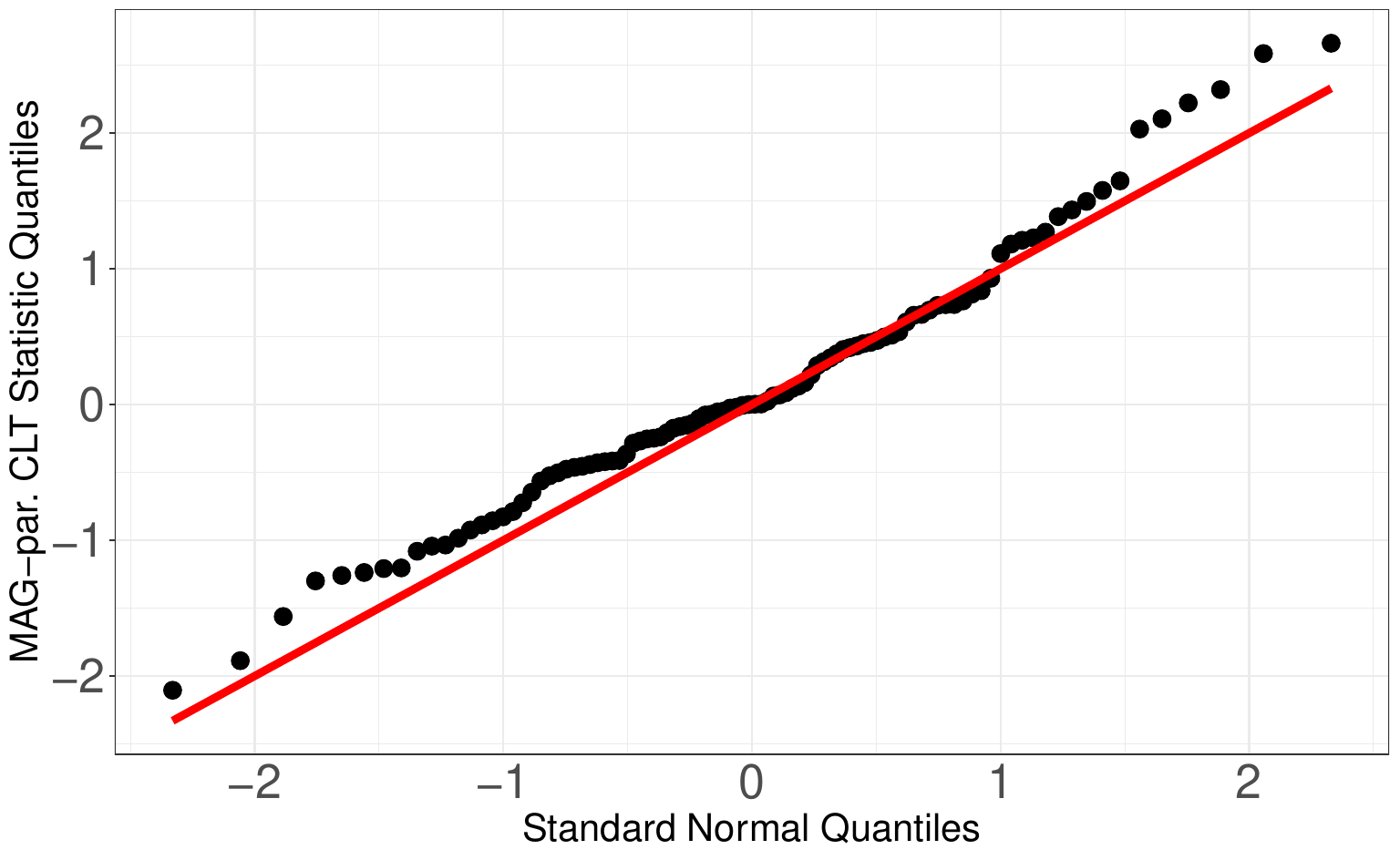}
\caption{QQ-plot for the CLT statistic of the AR and MAG-parameters of a MAGMAR$(1,1)$ model. The quantiles are obtained by 1) simulating $m=100$ MAGMAR$(1,1)$ processes of length $n=1000$, 2) estimating the parameters $m=100$ times, 3) estimating the standard deviation of the estimator based on the $m=100$ estimates 4) calculating the CLT statistic $Z_{i,n} = \frac{\hat{\phi}_i - \phi_0}{\text{sd}(\hat{\phi})}$, where $i = \in \{1,\hdots,m\}$ 5) calculating the sample quantiles.}
\label{Fig:QQplot_MAGMAR(1,1)}
\end{figure}
\subsection{Modeling Toy Example}
\label{Subsection:Simulation_Study:Toy_Example}
\noindent
\SPP{In this section we exemplify our modeling approach for the MAGMAR-Copula model. To this end we simulate $n = 1000$ observations $\{\tilde{U}_t\}_{t = 1}^n \sim \Psi$-MAGMAR$(1,1)$-(g)-(t) with $\phi = 1.4$ and $\theta = (0.3, 4)$. and quantile transform to $Y_t = \Phi^{-1}(\tilde{U}_t)$ for all $t$. Then we follow the modeling approach laid out in Sect.~\ref{Subsubsection:MAGMAR:Distributional:Joint_distribution_Likelihood}.
The exemplary time series is given in the upper panel of Fig.~\ref{Fig:SimI_ts_residual_plot}. Now we treat the time series as if we do not know the true data generating process and try to model it with the MAGMAR-copula model. As the first step, we transform the time series to pseudo-observations by using the empirical CDF, yielding the time series $\{V_t\}_{t=1}^n$. Second, we run automated model selection of the model orders $p$ and $q$ using the \texttt{auto.arima} command from the \texttt{forecast}-package in \texttt{R}, see \cite{hyndman2008automatic}. This yields $p=q=1$.
%Third we investigate the scatter plot of consecutive pseudo-observations, Fig.~\ref{Fig:SimI_scatterAR1}, and 
Third we perform automatic model selection for the copula of $(U_t, U_{t-1})$ to get an idea of suitable copulas. Model selection by BIC yields a Gumbel copula. We note that this copula is distinct from the AR-copula, however we can still use it as starting point. Model selection by AIC yields a $t$-copula for $(U_t, U_{t-1})$. Hence, we proceed to fit an adjusted-MAGMAR$(1,1)$-(g)-(t) and an adjusted-MAGMAR$(1,1)$-(t)-(t) model to the data. We recall that the purpose of the adjustment transformation $\Psi$ is to ensure that the resulting process has a standard uniform stationary distribution. For this we follow the procedure proposed in Sect.~\ref{Subsubsection:MAGMAR:Distributional:Joint_distribution_Likelihood} and estimate the parameters and the adjustment transformation iteratively with two iterations.
Comparison of information criteria then shows that the $\Psi^{(2)}$-MAGMAR$(1,1)$-(g)-($t$) model is more suitable. In the lower panel of Fig.~\ref{Fig:SimI_ts_residual_plot} the residuals of this model are plotted. We observe that the MAGMAR$(1,1)$-(g)-($t$) model successfully captures the dynamics.
\begin{figure}
\centering
\includegraphics[scale=0.5]{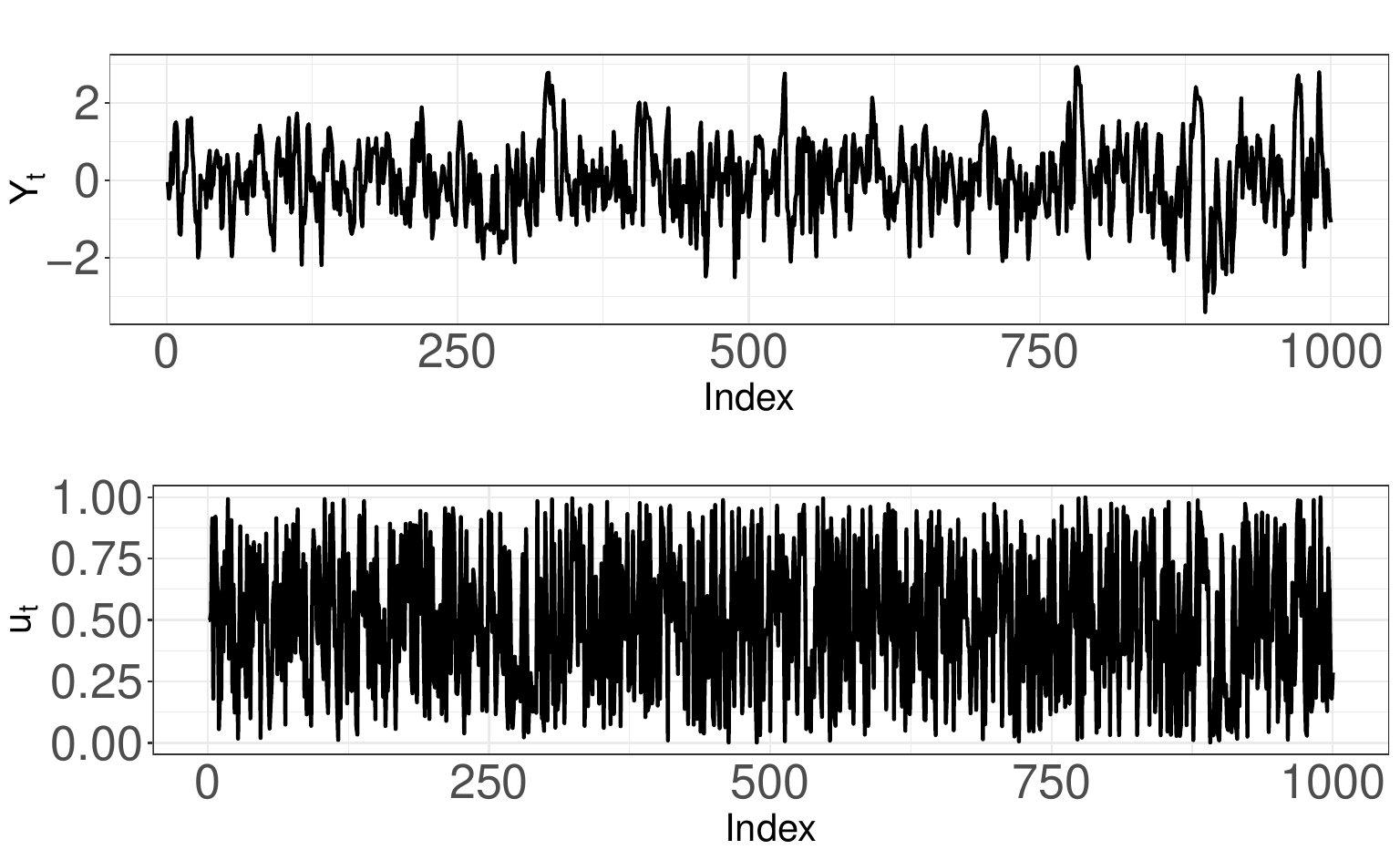}
\caption{The simulated time series $\{Y_t\}$ and residuals of a MAGMAR$(1,1)$ model fitted to the time series.}
\label{Fig:SimI_ts_residual_plot}
\end{figure}
}%TODO Check parameter values.
\subsection{Comparison with other Non-Markovian Copula-Based Models}
\label{Subsection:Simulation_Study:Comparison}
\noindent
\SPP{We compare the proposed model with Markovian copula-based time series and the infinite partial dependence model by \cite{mcneil2022time}. We use the simulated time series from Sect.~\ref{Subsection:Simulation_Study:Toy_Example} and fit these models to the data. For the standard copula-based time series, we perform order selection again using the \texttt{auto.arima} command (cf. \cite{hyndman2008automatic}). This yields $p=3$. We then select the copula of $(U_t, U_{t-1})$ automatically by BIC, calculate the conditional variables $(U_t|U_{t-1})$ and $(U_{t-2}|U_{t-1})$ and proceed with automatic model selection.
%This yields the following copulas. AR$(1)$: Gumbel, AR$(2)$: normal, AR$(3)$: independence. Since the AR$(3)$ copula is chosen as independence,
We also take the fourth autoregressive copula into account.
%, which is chosen as $270^{\circ}$-rotated Clayton copula.
The AIC of these two models is $-669.4702$ and $-667.4702$, respectively. The model by \cite{mcneil2022time} offers two potential versions for the given time series. First, replacing every copula in the infinite D-vine by a Gumbel copula, or taking all copulas as normal and only replacing a fixed amount by suitable copulas. We test both version obtaining AIC values of $-623.3599$ and $-677.0362$, respectively. For the MAGMAR-model we follow the procedure described above in Sect.~\ref{Subsection:Simulation_Study:Toy_Example} and obtain an AIC of $-819.0715$. This is expected since the DGP is the (adjusted) MAGMAR-copula model. Nevertheless, this simulation study confirms that the MAGMAR-copula model can produce complicated temporal dependencies, which can only be properly modeled using the correct model specification. Furthermore the simulation study showcases that a sparse model such as the $\Psi^{(2)}$-MAGMAR$(1,1)$ model in fact performs better than copula-based models with high autoregressive orders.}
\section{Application to US Inflation}
\label{Section:Application}
\noindent In this section, the modeling performance of the MAGMAR-Copula model is investigated. We compare its performance with the performance of the copula-based time series model with infinite partial dependence proposed by \cite{mcneil2022time}. Just as in \cite{mcneil2022time}, quarterly US inflation is modeled. We use the consumer price index data, \texttt{cpi} from the \texttt{R} package \texttt{tscopula} provided by the authors, and calculate the inflation as the growth rates of the \SPP{\texttt{cpi}.} The time series comprises 244 observations spanning from 1960Q1 to 2020Q4. For implementation the \texttt{R} package \texttt{rvinecopulib} by \cite{nagler2025rvinecopulib} is used. The package allows for fast vine-copula modeling using \texttt{C++} via \texttt{rcpp}, see \cite{eddelbuettel2011rcpp}.
\SPP{We consider two versions of the MAGMAR-Copula model. First the pure model without any adjustment transformation. Second, we consider the adjusted model, where we estimate the adjustment transformation and the parameters in an iterative way as described in Sect.~\ref{Subsubsection:MAGMAR:Distributional:Joint_distribution_Likelihood}
For this model we add the prefix '$\Psi^{(k)}$-, where $k\in\mathbb{N}$ denotes the number of iterations.}    
In order to estimate the adjustment transformation $\Psi$, the \texttt{R} package \texttt{kde1d} by \cite{nagler2025kde1d} is used. 
Before estimating the respective models, the time series is transformed to pseudo-observations by applying the rescaled empirical CDF to the data. The models' performances are assessed with the AIC and BIC. The results can be found in Table~\ref{AIC_BIC}. The models from \cite{mcneil2022time} \SPP{are} called 'IPD-name' \SPP{if all copulas are chosen as the same copula.} 'Name' is a placeholder for the copula, e.g. 'Gumbel'. \SPP{If only certain copulas are replaced by non-Gaussian copulas, we will follow the same naming scheme as for the MAGMAR-copula model. The results for the IPD-normal and IPD-Gumbel model are directly taken from their paper.} The MAGMAR-copula models follow the naming-scheme from Sect.~\ref{Subsubsection:MAGMAR:Distributional:Joint_distribution_Likelihood}. 
\begin{table*}
\caption{Number of Parameters, negative log-likelihood and Information criterion values of the mentioned models fitted to quarterly US inflation data (244 observations). \SPP{The 'IPD-' models} are the models for infinite partial copula dependence from \cite{mcneil2022time}. The MAGMAR$(p,q)$ model is the proposed model. The letters after the model orders indicate which copulas are used at the respective order, as is described in Sect.~\ref{Subsubsection:MAGMAR:Distributional:Joint_distribution_Likelihood}. The prefix '$\Psi^{(k)}$-' indicates that the adjustment transformation is conducted and $k$ iterations of estimation were performed. The highlighted values are the lowest ones in the respective column.}
\centering
\label{AIC_BIC}
\begin{tabular}{@{}lrrrc@{}}
\hline
Model & \# pars & nLL & AIC & BIC 
\\
\hline
\hline
IPD-normal & $6$ & $-98.31$ & $-184.62$ & $-163.64$
\\
IPD-Gumbel & $6$ & $-110.64$ & $-209.28$ & $-188.30$
\\
IPD-(j,t,n,g) & $5$ & $-111.12$ & $-210.25$ & $-189.26$
\\
\hline
MAGMAR(4, 0)-(g,g,t,g) & $5$ & $-111.05$ & $-212.10$ & $-194.61$
\\
MAGMAR(5, 0)-(g,g,t,g,g) & $6$ & $-111.27$ & $-210.54$ & $-189.56$
\\
\hline
MAGMAR(4, 1)-(n,n,n,n)-(n) & $5$ & $-93.48$ & $-176.95$ & $-159.47$
\\
MAGMAR(4, 1)-(g,i,n,g)-(t) & $5$ & $-112.16$ & $-214.31$ & $-196.83$
\\
$\Psi^{(2)}$-MAGMAR(4, 1)-(n,n,n,n)-(n) & $5$ & $-103.74$ & $-197.47$ & $-179.99$
\\
%$\Psi$-MAGMAR(4, 1)-(c180,i,n,c180)-(t) & $5$ & $-93.74$ & $-177.48$ & $-160.00$
%\\
$\Psi^{(2)}$-MAGMAR(4, 1)-(g,i,n,g)-(t) & $5$ & $\mathbf{-116.80}$ & $\mathbf{-223.60}$ & $\mathbf{-206.12}$
\\
\hline
\end{tabular}
\end{table*}
\noindent The $\Psi^{(2)}$-MAGMAR(4,1)-(g,i,n,g)-(t) model minimizes the AIC, as well as the BIC. The necessity of including the Gumbel copula indicates asymmetric temporal dependence while the inclusion of the $t$-copula indicates heavy-tailedness of the dependence structure. All of the best performing models incorporate 4 autoregressive orders. This result \SPP{could indicate stochastic} seasonality in the quarterly inflation. The necessity of including a MAG part indicates that the data generating process of US inflation may not fulfill the Markov assumption. \SPP{This is also reflected by the fact that increasing the model order to $p=5$ in the Markovian models does not enhance the fit.} \SPP{Furthermore, we observe that the adjusted MAGMAR-copula models perform better than their unadjusted counterparts. This is expected since the model is misspecified without adjustment.} Our proposed model minimizes the AIC/BIC in comparison with the IPD model by \cite{mcneil2022time}. \SPP{This could indicate that the long-term temporal dependence structures have a complex form, which are best captured by the MAGMAR-model.} However, we note that the model by \cite{mcneil2022time} is still competitive. \SPP{In Fig.~\ref{Fig:Inflation_MAGMAR_fit} we show the in-sample fit of the best-performing model to the data. Since the model naturally produces probabilistic in-sample predictions, we report the conditional median and the conditional 90\%-confidence interval. The fit indicates that the $\Psi^{(2)}$-MAGMAR(4, 1)-(g,i,n,g)-(t) model seems to capture the dynamics of US inflation quite well. In Fig.~\ref{App:Fig:Inflation_residual} in \ref{App:Section:Additional_Plots} we also provide the residual time series transformed to the $z$-scale and its ACF/PACF. There seems to be not much structure left, which confirms that the model captures the dynamics well.}
\begin{figure}
\centering
\includegraphics[scale=0.5]{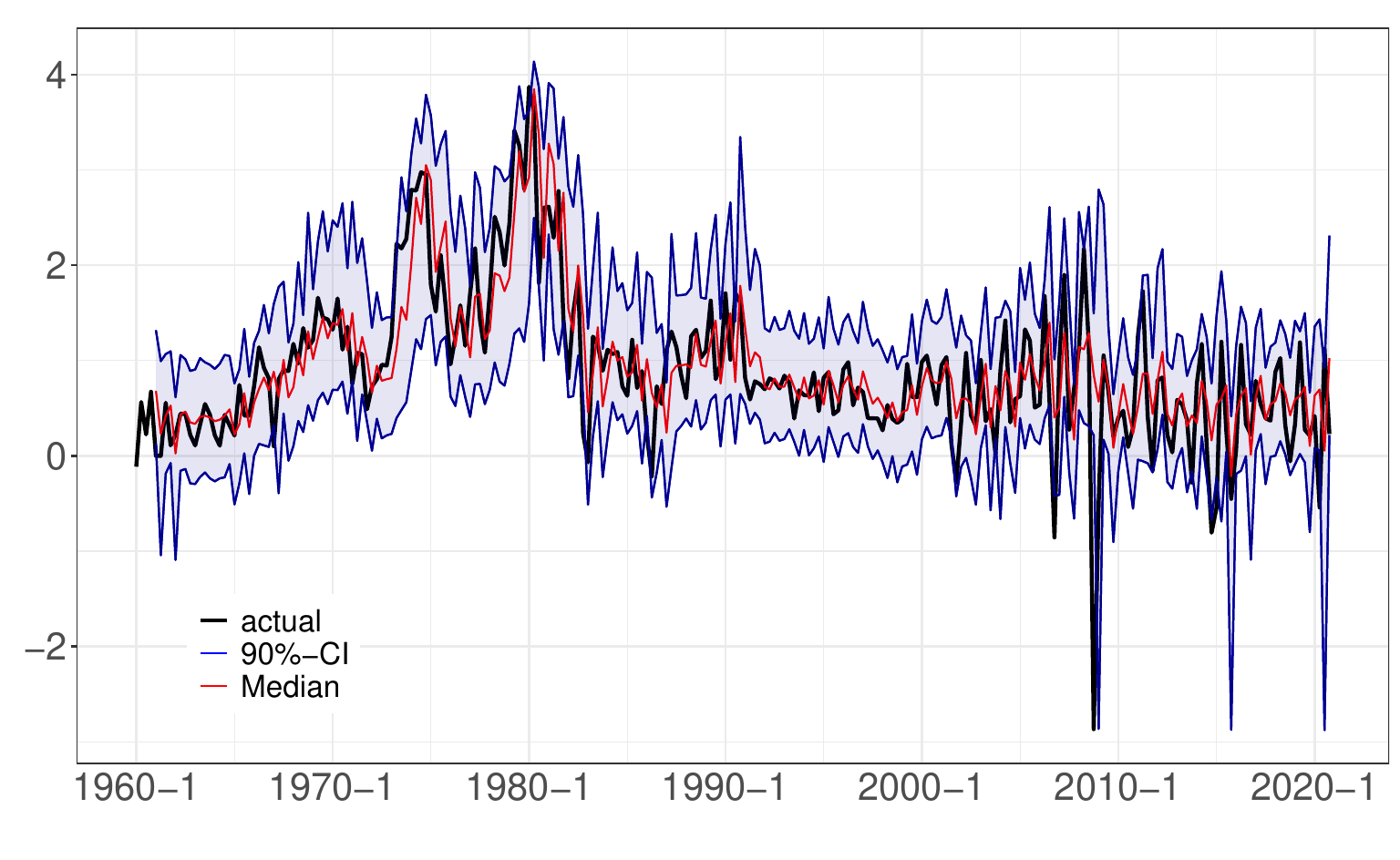}
\caption{In-sample fit of the KDE-quantile-transformed $\Psi^{(2)}$-MAGMAR-(g,i,n,g)-(t) model to US inflation.}
\label{Fig:Inflation_MAGMAR_fit}
\end{figure}
\section{Conclusion}
\label{Section:Conclusion}
\noindent Non-linear temporal dependencies as well as deviations from normal or elliptical distributions are apparent in many univariate time series. A popular approach to capture these effects is to employ a copula $C_{\phi}$ for the temporal dependence structure equipped with a suitable marginal distribution $F_X$. The resulting model is the well-known copula-based time series model. Alongside stationarity, the model assumes a finite Markov order $p$. This means that the distributional properties of a random variable of the time series can only be conditionally influenced by the preceding $p$ values. This results in the model not being able to properly capture long-term autoregressive effects, e.g. persistency in the mean or the variance. Properly modeling persistency in the mean and the variance was a major improvement for 'classical' time series modeling. To this end, the moving average term (MA) was introduced to autoregressive (AR) models resulting in the ARMA model. The generalization of the autoregressive conditional heteroskedasticity models (ARCH) was introduced to account for persistency in the variance, i.e. volatility clustering. Inspired by this approach we introduce a moving aggregate (MAG) part to the model equation of copula-based time series models. The moving aggregate is a non-linear function of the current and past innovations. The resulting model is the Moving Aggregate Modified Autoregressive Copula-Based Time Series Model (MAGMAR-Copula time series model). After introducing the model we discuss the model components and nested models.
%The MAGMAR-Copula model can be understood as a generalization to classical copula-based time series models and the classical ARMA model since both models are nested.
%In comparison with the copula-ARMA model from \cite{mcneil2022time} our model is in general not able to give explicit forms for partial dependencies, while the copula-ARMA model is able to do so. The main advantage of our model is that it is able to incorporate copulas with an arbitrary amount of parameters without any effort. This is especially relevant for the $t$-copula and related copulas.
We derive both the MAG($\infty$) and the AR($\infty$) representations and continue to derive distributional properties such as the conditional marginal distribution and the likelihood function. Furthermore, we prove that the MAGMAR$(1,1)$-Copula model is stationary when the AR-copula fulfills a contraction-on-average property. \SPP{We exemplify the modeling approach in a simulation study. In particular, we suggest to incorporate an adjustment transformation that guarantees that the models marginal distribution is well specified.}
In an application to US inflation we compare our model with the model for infinite order partial dependence of \cite{mcneil2022time}. Their model is also able to capture the infinite memory of a time series \SPP{using copulas.} The comparison in the application is conducted by using information criteria (AIC and BIC). We find that our model is competitive.\\
%Multivariate extension
\SPP{Although our proposed MAGMAR-Copula model is quite general, it is a model for univariate time series. The incorporation of other variables as covariates is in principle possible but may not necessarily enhance the performance of the model in its current form. For the modeling of multivariate time series, we suggest to use copula-GARCH type-models and employ MAGMAR-Copulas for temporal modeling instead of ARMA-GARCH models (\cite{hu2006dependence, jondeau2006copula, aloui2014dependence} are standard references for the copula-GARCH model while \cite{berrisch2023modeling} is a recent application).}
%Forecasting
Furthermore, we would like to explore the forecasting performance of our model. As mentioned earlier, our model is able to capture volatility clusters. It will be interesting to see how our model performs in comparison with ARCH-models. From a theoretical point of view, further distributional properties are of interest. \SPP{Especially interesting are asymptotic normality for the MLE of MAG$(1)$ processes and the stationarity of higher order processes. Further exploring theoretical properties of the common-innovations model by \cite{joe2014dependence} and also testing it in simulations and on real data is of high interest.}

%Especially interesting are the joint distributions of consecutive random variables in the time series, $U_t$ and $U_{t-1}$. %This joint distribution is related to the stationary distribution of the time series, which is also of high interest.

%\section*{Acknowledgement}
%\noindent
%This research was (partially) funded in the course of TRR 391 \textit{Spatio-temporal Statistics for the Transition of Energy and Transport} (520388526) by the Deutsche Forschungsgemeinschaft (DFG, German Research Foundation).

\section*{Data Availability Statement}
\noindent
The data is publicly available through the \texttt{R} package \texttt{tscopula}.

%\section*{Conflict of Interest Statement}
%\noindent
%The author declares that there is no conflict of interest.

\section*{Funding Statement}
\noindent
This research was (partially) funded in the course of TRR 391 \textit{Spatio-temporal Statistics for the Transition of Energy and Transport} (520388526) by the Deutsche Forschungsgemeinschaft (DFG, German Research Foundation).

\section*{Associated Key Words}
\noindent
Copula, Copula-based time series, Dependence Modeling, Non-linear time series, Persistency, Vine Copula.

\section*{Mathematics Subject Classification}
\noindent
62-02, 62H05, 62M09, 62M10, 62P20, 60G10

\bibliography{bibi.bib}

\appendix

\section{Supporting Results}
\label{App:Section:Supporting_Results}
\subsection{Likelihood Calculation for MAGMAR$(p,q)$}
\label{App:Subsection:Additional:Likelihood_Iteration}
The following calculation steps have to be performed to evaluate the likelihood of the MAGMAR$(p,q)$-Copula model.
\begin{itemize}
\item[1)] Set initial values for the first $s = \max\{p,q\}$ observations. (A natural choice for initial values in this context is 0.5 because for symmetric marginal distributions this corresponds to the expectation value. This is also the choice of initial values that we choose for our empirical application.)
\item[2)] For $i \in \{s+1,\hdots,T\}$ calculate $R_{\phi_1, \hdots, \phi_p}^1((u_{i-1},\hdots,u_{i-p}), u_i)$ iteratively.
\item[3)] For $i \in \{s+1,\hdots,T\}$ calculate
\begin{align*}
w_i = (R_{\theta_1, \hdots, \theta_q}^1)((w_{i-1}, \hdots, w_{i-q}), R_{\phi_1, \hdots, \phi_p}^1((u_{i-1},\hdots,u_{i-p}), u_i))
\end{align*}
iteratively.
\item[4)] For $i \in \{s+1,\hdots,T\}$ calculate $f_{U_i|\mathcal{F}_{i-1}}(u_i)$ (Eq.~\ref{ConditionalDensityMAGMARpq} for MAGMAR$(p,q)$ and Eq.~\ref{MAGMAR11ConditionalDensity} for MAGMAR$(1,1)$).
\item[5)] Take the product of all conditional densities for the likelihood or the negative sum of log-densities for the negative log-likelihood: $L = \prod_{i = s+1}^T f_{U_i|\mathcal{F}_{i - 1}}(u_i)$.
\end{itemize}

\subsection{AR and MAG representation of the MAGMAR(p,q) model}
\label{App:Subsection:AR_and_MAG_Representation}
\noindent The AR representation of the MAGMAR$(p,q)$-model (Eq.~\ref{Eq:MAGMAR_p_q}) can be derived by substituting for all $\{w_{t-i}\}_{i=1}^{\infty}$ iteratively. We start with the following expression:
\begin{align*}
w_t = R_{\theta_1, \hdots, \theta_q}^1 \left( (w_{t-1}, \hdots, w_{t-q}), R_{\phi_1, \hdots, \phi_p}((U_{t-1}, \hdots, U_{t-p}), U_t) \right).
\end{align*}
Now, $w_{t-1}, \hdots, w_{t-q}$ have to be substituted again. Formally this amounts to composing the function $R_{\theta_1, \hdots, \theta_q}^1(\circ, R_{\phi_1, \hdots, \phi_p}^1((U_{t-1}, \hdots, U_{t-p}), U_t): (0,1)^q \rightarrow (0,1)$ with a vector of functions:
\begin{align*}
w_t =& R_{\theta_1, \hdots, \theta_q}^1(\circ, R_{\phi_1, \hdots, \phi_p}^1((U_{t-1}, \hdots, U_{t-p}), U_t)
\\
&\circ
\begin{pmatrix}
R_{\theta_1, \hdots, \theta_q}^1(\circ, R_{\phi_1, \hdots, \phi_p}^1((U_{t-2}, \hdots, U_{t-p-1}), U_{t-1})
\\
\vdots
\\
R_{\theta_1, \hdots, \theta_q}^1(\circ, R_{\phi_1, \hdots, \phi_p}^1((U_{t-q-1}, \hdots, U_{t-p-q}), U_{t-q})
\end{pmatrix}
\circ \hdots
\end{align*}
Now the $\{w_{t - i}\}_{i = 2}^{\infty}$ in each vector has to be substituted for again. The input for each element of the vector is again a $q$-dimensional vector. The total dimension of the input vector at the third level hence is given by $q^2$. We shorten notation by introducing the $(\stackrel{\rightarrow}{\bigcirc}_d)_{i \in \mathbb{N}_0}$ symbol which acts as follows for a sequence of functions $f_i: \mathbb{R}^d \rightarrow \mathbb{R}$, $i \in \mathbb{N}_0$
\begin{align*}
(\stackrel{\rightarrow}{\bigcirc}_d)_{i \in \mathbb{N}_0} f_i(\circ) = f_0
\circ
\begin{pmatrix}
f_1
\\
\vdots
\\
f_d
\end{pmatrix}
\circ
\begin{pmatrix}
f_{d + 1}
\\
\vdots
\\
\vdots
\\
f_{d + d^2}
\end{pmatrix}
\circ \hdots
\end{align*}
Using the newly introduced notation, the AR representation may be written as
\begin{align}
w_t = (\stackrel{\rightarrow}{\bigcirc}_q)_{i \in \mathbb{N}_0} R^1_{\theta_1, \hdots, \theta_q}(\circ, R_{\phi_1, \hdots, \phi_{p}}((U_{t - 1 - i}, \hdots, U_{t - p - i}), U_{t-i})
\end{align}
Using the same notation, the MAG($\infty$)-representation can be written as
\begin{align*}
U_t =& R_{\phi_1, \hdots, \phi_p}^{-1} (\circ, R_{\theta_1, \hdots, \theta_q}^{-1} ( (w_{t-1}, \hdots, w_{t-q}), w_t))
\\ 
& \circ
\begin{pmatrix}
R_{\phi_1, \hdots, \phi_p}^{-1} (\circ, R_{\theta_1, \hdots, \theta_q}^{-1} ( (w_{t-2}, \hdots, w_{t-q-1}), w_{t - 1}))
\\
\vdots
\\
R_{\phi_1, \hdots, \phi_p}^{-1} (\circ, R_{\theta_1, \hdots, \theta_q}^{-1} ( (w_{t-p-1}, \hdots, w_{t-p-q}), w_{t - p}))  \end{pmatrix}
\circ \hdots
\\
=& (\stackrel{\rightarrow}{\bigcirc}_p)_{i \in \mathbb{N}_0}
R_{\phi_1, \hdots, \phi_p}^{-1} (\circ, R_{\theta_1, \hdots, \theta_q}^{-1} ( (w_{t-1-i}, \hdots, w_{t-q-i}), w_{t-i}))
\end{align*}
\subsection{Conditional Distribution MAGMAR(p,q)}
\label{App:Subsection:Conditional_Density_MAGMAR(p,q)}
\noindent The conditional distribution of $U_t$ given all past observations is given by:
\begin{align}
f_{U_t|\mathcal{F}_{t-1}}(x) =& c_{\theta_q}\left[ R_{\phi_1, \hdots, \phi_p}^1((u_{t-1}, \hdots, u_{t-p}), x), R^1_{\theta_1, \hdots, \theta_{q-1}}((w_{t-2}, \hdots, w_{t-q}), w_{t-1})  \right] \nonumber
\\
& \prod_{i = 0}^{p - 1} c_{\phi_{p - i}} \bigg[ R^1_{\phi_1, \hdots, \phi_{p - 1 - i}}((u_{t-1}, \hdots, u_{t-p+1+i}), x), \nonumber
\\
& R^2_{\phi_1, \hdots, \phi_{p - 1 - i}}((u_{t-1}, \hdots, u_{t-p+1+i}), u_{t-p+i}) \bigg].
\label{ConditionalDensityMAGMARpq}
\end{align}
The proof is analog to the proof for the MAGMAR(1,1)-Copula case.
\subsection{Joint Conditional Distribution of the MAGMAR model}
\label{App:Subsection:JointConditionalDistribution}
\noindent The joint conditional distribution of $(U_t, U_{t-1})|\mathcal{F}_{t-2}$ is derived. The density transform theorem for multivariate transformations \citep[see][5.6]{devore2012modern} is utilized. The two random variables $w_t$ and $w_{t-1}$ are transformed by the function $G$ to the random variables $U_t$ and $U_{t-1}$. The appropriate choice for $G$ and the corresponding inverse and Jacobian are given by,
\begin{align*}
\begin{pmatrix}
z
\\
w
\end{pmatrix}
:=&
G( (a, b)^T ) =
\begin{pmatrix}
h_{\phi}^{-1}(h_{\theta}^{-1}(a, w_{t-2}), u_{t-2})
\\
h_{\phi}^{-1}(h_{\theta}^{-1}(b, a), h_{\phi}^{-1}(h_{\theta}^{-1}(a, w_{t-2}), u_{t-2}))
\end{pmatrix},
\\
=&
\begin{pmatrix}
h_{\phi}^{-1}(h_{\theta}^{-1}(a, w_{t-2}), u_{t-2})
\\
h_{\phi}^{-1}(h_{\theta}^{-1}(b, a), z)
\end{pmatrix}.
\\
\begin{pmatrix}
a
\\
b
\end{pmatrix}
:=&
G^{-1}(z,w) =
\begin{pmatrix}
h_{\theta}(h_{\phi}(z, u_{t-2}), w_{t-2})
\\
h_{\theta}(h_{\phi}(w, z), h_{\theta}(h_{\phi}(z, u_{t-2}), w_{t-2})) 
\end{pmatrix},
\\
=&
\begin{pmatrix}
h_{\theta}(h_{\phi}(z, u_{t-2}), w_{t-2})
\\
h_{\theta}(h_{\phi}(w, z), a) 
\end{pmatrix}
\\
\mathcal{J} =&
\begin{pmatrix}
c_{\theta}[h_{\phi}(z, u_{t-2}), w_{t-2}] c_{\phi}[z, u_{t-2}] & 0
\\
\partial_z (\hdots) & c_{\theta}[h_{\phi}(w,z), h_{\theta}(h_{\phi}(z,u_{t-2}), w_{t-2})] c_{\phi}[w, z]
\end{pmatrix},
\\
|\mathcal{J}| =& c_{\theta}[h_{\phi}(z, u_{t-2}), w_{t-2}] c_{\phi}[z, u_{t-2}] c_{\theta}[h_{\phi}(w,z), h_{\theta}(h_{\phi}(z,u_{t-2}), w_{t-2})] c_{\phi}[w, z].
\end{align*}
Here $u_{t-2}$ and $w_{t-2}$ ought to be understood as (deterministic by conditioning) constants. Since $w_t$ and $w_{t-1}$ are independent and uniformly distributed their contribution to the conditional density is just $1_{(0,1)^2}$. The density can be formulated as
\begin{align*}
&f_{U_t, U_{t-1}|\mathcal{F}_{t-2}}(z, w) = f_{w_t, w_{t-1}}(z, w) |\mathcal{J}| 
\\
=& c_{\theta}[h_{\phi}(z, u_{t-2}), w_{t-2}] c_{\phi}[z, u_{t-2}] c_{\theta}[h_{\phi}(w,z), h_{\theta}(h_{\phi}(z,u_{t-2}), w_{t-2})] c_{\phi}[w, z].
\end{align*}
The partial conditional density of $U_t, U_{t-k}|U_{t-1},\hdots,U_{t-k+1},U_{t-k-1},\hdots$ can again be derived by using the density transformation theorem. The random variables $w_t$ and $w_{t-k}$ are transformed to the random variales $U_t$ and $U_{t-k}$. The appropriate choice of the transformation $G$ is,
\begin{align*}
G((a,b)^T) =&
\begin{pmatrix}
z
\\
w
\end{pmatrix}
= 
\begin{pmatrix}
h_{\phi}^{-1}(h_{\theta}^{-1}(a, w_{t-k-1}), u_{t-k-1})
\\
h_{\phi}^{-1}(h_{\theta}^{-1}(b, w_{t-1}(a)), u_{t-1})
\end{pmatrix}
\\
G^{-1}(z,w) =&
\begin{pmatrix}
a
\\
b
\end{pmatrix}
=
\begin{pmatrix}
h_{\theta}(h_{\phi}(z, u_{t-k-1}), w_{t-k-1})
\\
h_{\theta}(h_{\phi}(w, u_{t-1}), w_{t-1}(z))
\end{pmatrix}
\\
\mathcal{J} =& 
\begin{pmatrix}
c_{\theta}[h_{\phi}(z, u_{t-k-1}), w_{t-k-1}) c_{\phi}[z, u_{t-k-1}] & 0
\\
\partial_{z} (\hdots) & c_{\theta}[h_{\phi}(w, u_{t-1}), w_{t-1}(z)] c_{\phi}[w, u_{t-1}]
\end{pmatrix}
\\
|\mathcal{J}| =& c_{\theta}[h_{\phi}(z, u_{t-k-1}), w_{t-k-1}) c_{\phi}[z, u_{t-k-1}] c_{\theta}[h_{\phi}(w, u_{t-1}), w_{t-1}(z)] c_{\phi}[w, u_{t-1}].
\end{align*}
Here $w_{t-1}(z)$ means evaluating the AR($k-1$) representation of $w_{t-1}$ at $u_{t-1},\hdots,u_{t-k+1}$ and $w_{t-k} = h_{\theta}(h_{\phi}(z, u_{t-k-1}), w_{t-k-1})$ and $U_{t-k} = z$.
\subsection{Verifying conditions for stationarity for several copulas.}
\label{App:Subsection:ContractionCopulas}
\noindent
Here we verify the assumptions from \cite{douc2014nonlinear} to establish the stationarity and ergodicity of the MAGMAR$(1,1)$ time series for the normal, $t$ and Clayton Copula. For the sake of completeness, we give the assumptions from \cite{douc2014nonlinear} here too. The assumptions revolve around a stochastic recurrence equation, $X_t = f_{Z_t}(X_{t-1})$. The assumptions are as follows:\\
\textbf{Assumption A4.36:} $Z_t$ is strictly stationary and ergodic.\\
\textbf{Assumption A4.37:} There is a measurable function $K_z$ such that $d(f_z(x), f_z(y)) \leq K_z d(x,y)$, with $\mathbb{E}(\log^+(K_{Z})) \leq \infty$ and $\mathbb{E}(\log(K_{Z})) < 0$\\
\textbf{Assumption A4.38:} There is a $x \in \mathbb{R}$ s.t. $\mathbb{E}(\log^+ (d(x, f_Z(x))) < \infty$.\\
In order to use the result from \cite{douc2014nonlinear}, we recognize that the MAGMAR$(1,1)$ updating equation can be interpreted as stochastic recurrence equation. The MAGMAR$(1,1)$ updating equation may be written as $U_t = f_{Z_t}(U_{t-1})$, where $Z_t$ is the stationary and ergodic MAG$(1)$ process and $f_{z}(x) = h^{-1}_{\alpha}(z,x)$. Now we verify the conditions.
\subsubsection{Normal copula as AR-Copula}
\noindent
Let $\{U_t\}$ be a MAGMAR$(1,1)$-Copula time series where the AR-Copula is the normal Copula with parameter $|\alpha| < 1$. The conditions for the theorem are \textbf{not} fulfilled for $\kappa = \text{id}$. However, we are able to show that there is a mapping $\kappa:\mathbb{R} \rightarrow (0,1)$ such that the theorem is applicable to the time series $\{Y_t\}$ with $Y_t = \kappa^{-1}(U_t)$. The corresponding stochastic recurrence equation is specified by $\tilde{f}_z(x) = \kappa^{-1}(h_{\alpha}^{-1}(z,\kappa(x)))$. \footnote{
$U_t = f_{Z_t}(U_{t-1}) \iff Y_t = \kappa^{-1}[U_t] = \kappa^{-1}[f_{Z_t}(\kappa[\kappa^{-1}[U_{t-1}]]) = \kappa^{-1}[f_{Z_t}(\kappa[Y_{t-1}])$. \SPP{As an example, the reader may think of $\kappa$ as $\kappa = \Phi$, which will also be used in the following.}} Now we can check the assumptions.\\
\textbf{Assumption A4.36:} $Z_t = V_t = h^{-1}_{\theta}(w_t, w_{t-1})$ is stationary and $1$-dependent, and hence ergodic.\\
\textbf{Assumption A4.37:} We choose $d(x,y) = |x - y|$. Then by the mean value theorem, it follows that
\begin{align}
|\tilde{f}_{z}(x) - \tilde{f}_{z}(y)| \leq & |\max_{\zeta \in (x,y)} \tilde{f}_{z}'(\zeta)| |x - y|
\\
\leq & |\max_{\zeta \in \mathbb{R}} \tilde{f}_{z}'(\zeta)| |x - y|
\end{align}
where 
\begin{align*}
\tilde{f}_z'(\zeta) = \alpha \frac{\phi[\alpha\Phi^{-1}(\kappa(x)) + (1 - \alpha^2)^{\frac{1}{2}} \Phi^{-1}(z)]}{\phi[\Phi^{-1}(\kappa(x))]} \frac{\kappa'(x)}{\kappa'(\kappa^{-1}(\Phi(\alpha\Phi(\kappa(x)) + (1 - \alpha^2)^{\frac{1}{2}} \Phi^{-1}(z))}.
\end{align*}
We notice that for $\kappa = \Phi$ we obtain $\tilde{f}_z'(\zeta) = \alpha < 1$. Hence, choosing $\kappa = \Phi$ and $K_z = \alpha$ satisfies the requirement.\\
\textbf{Assumption A4.38:} There is a $x \in \mathbb{R}$ s.t. $\mathbb{E}(\log^+ (d(x, \tilde{f}_Z(x))) < \infty$. Choose e.g. again $\kappa = \Phi$, then $\mathbb{E}(\log^+ (d(x, \tilde{f}_Z(x))) = \log^+(\alpha) < \infty$.\\
From $\{Y_t\}$ being stationary it follows that $\{\kappa(Y_t)\} = \{U_t\}$ is also stationary.
\subsubsection{$t$-copula as AR-copula:}
\noindent
Let $C_{\alpha}$ be given as the $t$-copula with parameters $\alpha = (\rho, \nu)$. The inverse $h$-function is given as (cf. \cite{smith2010modeling})
\begin{align*}
h^{-1}_{\rho, \nu}(z,x) = t_{\nu} \bigg( t_{\nu + 1}^{-1}(z)\bigg[ \frac{(\nu + (t_{\nu}^{-1}(x))^2)(1 - \rho^2)}{\nu + 1} \bigg]^{\frac{1}{2}} + \rho t_{\nu}^{-1}(x) \bigg)
\end{align*}
Choosing $\kappa = t_{\nu}$, $\tilde{f}_z$ is given as
\begin{align*}
\tilde{f}_z(x) = t_{\nu+1}^{-1}(z)\bigg[ \frac{(\nu + x^2)(1 - \rho^2)}{\nu + 1} \bigg]^{\frac{1}{2}} + \rho x 
\end{align*}
\textbf{Assumption A4.37:} We \SPP{again write} $|\tilde{f}_z(x) - \tilde{f}_z(y)| \leq K_z |x - y|$ where
$K_z = \max_{\zeta \in \mathbb{R}} |\tilde{f}_z(x)'|$. The first derivative of $\tilde{f}_z$ is given as
\begin{align*}
\tilde{f}'_z(x) =  t_{\nu+1}^{-1}(z)\bigg[ \frac{(\nu + x^2)(1 - \rho^2)}{\nu + 1} \bigg]^{-\frac{1}{2}} \frac{x(1 - \rho^2)}{\nu + 1} + \rho.
\end{align*}
This expression is monotone increasing and obtains its absolute maximum at $x \rightarrow \pm \infty$. Hence $K_z = |t^{-1}_{\nu+1}(z) \bigg( \frac{1 - \rho^2}{\nu + 1} \bigg)^{\frac{1}{2}} + \rho|$. Then $\mathbb{E}(\log|K_{U_t}|)$ is the expectation of the logarithm of a folded unstandardized $t$-distribution. The expression can be evaluated numerically by e.g. Monte-Carlo simulation to show that $\mathbb{E}(\log |K_{V_t}|) < 0$. \SPP{In fact we observe that only for $\nu$ very close to zero, $0 < \nu < 0.2$, the expression is positive.} Hence part a) from condition A4.37 from \cite{douc2014nonlinear} is satisfied. Condition \SPP{b)} amounts to showing that $\mathbb{E}(\log^+ K_{V_t}) < \infty$. We may write $\mathbb{E}(\log^+ K_{V_t}) = \mathbb{E}(\log K_{V_t} \mathds{1}_{\{K_{V_t} > 1\}}) = \mathbb{E}(\log |Y \bigg( \frac{1 - \rho^2}{\nu + 1} \bigg)^{\frac{1}{2}} + \rho| \mathds{1}_{\{ Y > K_{\nu,\rho} \}})$ where $Y \sim t_{\nu + 1}$. Now there are two thresholds $s_1$ and $s_2$ depending on $\nu$ and $\rho$ such that we can write the expectation as $ \mathbb{E}(\log |Y \big( \frac{1 - \rho^2}{\nu + 1} \big)^{\frac{1}{2}} + \rho| \mathds{1}_{\{ Y > K_{\nu,\rho} \}}) = \int_{-\infty}^{s_1} \log[-y\big(\frac{1 - \rho^2}{\nu + 1} \big)^{\frac{1}{2}} + \rho] f_{t_{\nu+1}} dy + \int_{s_2}^{\infty}\log[y\big(\frac{1 - \rho^2}{\nu + 1} \big)^{\frac{1}{2}} + \rho] f_{t_{\nu+1}}(y) dy$. We observe that integrands are logarithmic, weighted by the $t$-distribution density. Since linear functions dominate the logarithm asymptotically and the first moment of the $t$-distribution exists for $\nu>2$, $\mathbb{E}(\log^+ K_{V_t})$ is finite as long as $\nu>2$.\\ \textbf{Assumption A4.38:} The assumption can be shown to be satisfied for e.g. $x=0$. Then $d(x,\tilde{f}_{z}(x)) = t_{\nu+1}^{-1}(z)\bigg[ \frac{\nu(1 - \rho^2)}{\nu + 1} \bigg]^{\frac{1}{2}}$. The existence follows again by the existence of the first moment of the $t$-distribution for $\nu > 2$.
\subsubsection{The Clayton Copula as AR-Copula:}
\noindent
The inverse $h$-function of the Clayton copula with parameter $\theta$ is given as,
\begin{align*}
h^{-1}(z,x) = ([z x^{\theta+1}]^{-\frac{\theta}{\theta + 1}} + 1 - x^{-\theta})^{-\frac{1}{\theta}}
\end{align*}
We choose $\kappa$ such that $\kappa(x) = x^{-\frac{1}{\theta}}$ and $\kappa^{-1}(x) = x^{-\theta}$ respectively. Then $\tilde{f}_z$ is given such that $\tilde{f}_z(x) = z^{-\frac{\theta}{\theta+1}} x + 1 - x = 1 + x(z^{-\frac{\theta}{\theta+1}} - 1)$. We can immediately see that an appropriate choice for $K_z$ is given as $K_z = |(z^{-\frac{\theta}{\theta+1}} - 1)|$.\\
\textbf{Assumption A4.37:} $\mathbb{E}(\log K_V) = \int_{0}^{1} \log |y^{-\frac{\theta}{\theta+1}} - 1| dy<\infty$ and $\mathbb{E}(\log^+ K_V)$ can be verified numerically.\\
\textbf{Assumption A4.38:} For $x = 0$, $\mathbb{E}(\log^+(d(x,f_Z(x))) = 1 < \infty$.

\subsection{The Autodependence of Markov and $q$-Dependent Processes}
\label{App:Subsection:Autodependence_Properties}
\noindent \SPP{Here we provide the proofs of the following two claims.
\begin{prop}
\label{Prop:Autodependence_Properties}
$\{X_t\}_{t\in\mathbb{Z}}$ is a Markov($p$) process $\iff$ for $r > p$
\begin{align}
F(x_0, x_r) =&\mathbb{P}(X_t \leq x_0, X_{t-r} \leq x_r|X_{t-1},\hdots,X_{t-r+1})
\\
=& \mathbb{P}(X_t \leq x_0|X_{t-1},\hdots,X_{t-r+1}) \times \mathbb{P}(X_{t-r} \leq x_r|X_{t-1},\hdots,X_{t-r+1})
\end{align}
\end{prop}
}
\SPP{This can be seen by considering the following scenario. Let $\{X_t\}_{t\in\mathbb{Z}}$ be a Markov($p$) process. Let $r > p$, then the partial dependence of $(X_t,X_{t-r}|X_{t-1},\hdots,X_{t-r})$ is independent: 
\begin{align}
F(x_0, x_r) =& \mathbb{P}(X_t \leq x_0, X_{t-r} \leq x_r|X_{t-1},\hdots,X_{t-r+1}) \nonumber
\\
=& \mathbb{P}(X_t \leq x_0|X_{t-1},\hdots,X_{t-r+1}, X_{t-r} \leq x_r) \times \mathbb{P}(X_{t-r} \leq x_r|X_{t-1},\hdots,X_{t-r+1}) \nonumber
\\
=& \mathbb{P}(X_t \leq x_0|X_{t-1},\hdots,X_{t-r+1}) \times \mathbb{P}(X_{t-r} \leq x_r|X_{t-1},\hdots,X_{t-r+1}) = F(x_0) \times F(x_r)
\end{align}
where in the first line we have used the law of total probability and in the second line that $\{X_t\}_{t\in\mathbb{Z}}$ is a Markov$(p)$ process and $r>p$, hence conditioning on $X_{t-r}$ can be dropped in the first probability measure. For the other direction, let $F_{X_{t},X_{t-r}|\mathcal{F}_{t-r+1}^{t-1}}$ be independent for all $r > p$. Then
\begin{align}
F_{X_t|\mathcal{F}_{t-r}^{t-1}}(x_0) =& \mathbb{P}(X_t \leq x_0|X_{t-1},\hdots,X_{t-r}) \nonumber
\\
=& \frac{\mathbb{P}(X_t \leq x_0|X_{t-1},\hdots,X_{t-r})\times \mathbb{P}(X_{t-r} \leq x_r|X_{t-1},\hdots,X_{t-r+1})}{\mathbb{P}(X_{t-r} \leq x_r|X_{t-1},\hdots,X_{t-r+1})}
\\
=& \frac{\mathbb{P}(X_t \leq x_0, X_{t-r} \leq x_r|X_{t-1},\hdots,X_{t-r+1})}{\mathbb{P}(X_{t-r} \leq x_r|X_{t-1},\hdots,X_{t-r+1})}
\end{align}
}
\SPP{\subsection{Conceptual Comparison to Similar Approaches}
\label{App:Subsection:Comparison}
\noindent We compare our proposed MAGMAR-model with the common-innovations copula model by \cite{joe2014dependence} and the infinite partial dependence model by \cite{mcneil2022time}. The common-innovations model has not yet been examined in detail, and there is no implementation. Therefore, some comparisons are not yet fully possible. A more in-depth examination and an implementation is a potential direction for future research. Nevertheless, we compare the models with regards to the following criteria:\\
\textbf{Flexibility}:
%All models are quite flexible given suitable choices of copulas. The feature that distinguishes the IPD-model from the other two models is that each a finite number of conditional copula for $(U_t, U_{t-k} | \mathcal{F}_{t-1:t-k+1})$ can be chosen. In the MAGMAR and the common-innovations approach, all conditional copulas result from the choice of AR and MAG-copulas (see \ref{App:Subsection:JointConditionalDistribution} for the exact form in MAGMAR-copula time series). This means that 
We discuss flexibility in terms of (partial) dependencies that the models can produce. That is, how flexible is each model in constructing $(U_t, U_{t-k} |\mathcal{F}_{t-1})$. The partial densities of the MAGMAR-model are derived in \ref{App:Subsection:JointConditionalDistribution}. Each partial density is a product of MAG and AR-copula densities. Given flexibility in these densities, the partial dependencies can be modeled flexibly too. Typically, however, a partial copula density of the (adjusted) MAGMAR-model will not be a 'standard' copula density. This means that complex dependencies can be modeled in arbitrary lags.
The IPD approach is also very flexible in modeling partial dependencies since they are modeled directly. For each $k \in \mathbb{Z}$, the copula of $(U_t, U_{t-k} |\mathcal{F}_{t-1:t-k+1})$ can be chosen. In practice, either only one copula family is allowed for all lags $k$ or all partial copulas are Gaussian except for a finite amount of non-Gaussian copulas for specified lags $k$ (see \cite{bladt2025semiparametric}). This approach can be arbitrarily flexible in the copula families. However, eventually, the partial copulas are Gaussian with increasing lag order $k$. Hence the main differences between the MAGMAR and the IPD model in terms of flexibility are that the long-range dependencies can differ. If long-range dependencies seem to approach Gaussianity and only some partial dependencies are non-linear, the IPD model is suitable. If long-range dependencies remain non-Gaussian (e.g. because of heavy-tailedness) the MAGMAR-Copula approach is suitable.
The common-innovations model is similar to the MAGMAR model and similar behavior can be expected.\\
\textbf{Computational demand:} The computational demand of all models is moderate with MAGMAR-copula models needing $1min$ on average to estimate a MAGMAR$(1,1)$ model with $1000$ observations. In order to estimate an adjusted MAGMAR$(1,1)$ model with two iterations, approximately $2min$ are necessary. The common-innovations model can be expected to have similar computational demand, however in this model no adjustment transformation is needed. The IPD model takes a few seconds for a similar estimation if all copulas are replaced. If only certain copulas are replaced the estimation takes around $15s$.\\
%In terms of computational demand, the IPD model dominates the MAGMAR and the CICO model, which both have similar demand. The computational demand of the CICO and the MAGMAR model is mainly caused by the need for iterative calculation of the likelihood. The reparametrization in the IPD model simplifies this problem. Furthermore, the IPD model can reasonably truncate. For moving average-like models (like the CICO and MAGMAR model) such truncation is not directly applicable. \\
%
\textbf{Modeling Approach:} The suggested modeling approach for the presented adjusted-MAGMAR copula model is to rely on descriptive statistics, such as autodependence and scatter plots.
%That is taking scatter plots of the observations with lagged observations into account and deciding for the AR-copulas based on visual inspection/dependence measures. While grid search w.r.t. say AIC is possible, the computational demand required is not reasonable.
A similar approach can be pursued for the IPD model. For the common-innovations model, however, it is hard to perform model selection directly from descriptive statistics. This is because the AR-model is latent. Approaches for model selection in the common-innovations model may be explored in future research.
%\textbf{Theoretical aspects:} Inference for IPD not clear. Well-behavedness based on classic ARMA, Inference for adjusted MAGMAR complicated due to transformation $\Psi$. Inference for CICO is most straight-forward because no transformation needed and no reliance on ARMA processes.\\
%
%\textbf{Simulation:} 1) Partial copulas of MAGMAR are only slightly disturbed by MAG part, Partial copulas in CICO are more disturbed. 2) Consistency of MAGMAR parameter estimation. 3) Time series with complicated and differing partial dependencies at each lag. MAGMAR should be best because AR-part can be taken\\
}

\section{Proofs}
\label{App:Section:Proofs}
\noindent
In this section we provide proofs for the propositions made in the manuscript.
\subsection{Proof of Proposition \ref{Prop:Relation_to_other_models}: Relation to other models}
\label{App:Subsection:Proof:Related models}
\begin{proof}
\SPP{We show the results step-by-step:}
\begin{itemize}
\item[i)] Let $X_t = \Phi^{-1}[ h_{\alpha}^{-1}(h_{\beta}^{-1}(w_t, w_{t-1}), \Phi(X_{t-1}))]$ for all $t$. We take the explicit form \SPP{of the (inverse) $h$-function associated to the normal copula from} \cite{smith2010modeling}: 
\begin{align*}
h_{\alpha}(u_1, u_2) =& \Phi ([\Phi^{-1}(u_1) - \alpha \Phi^{-1}(u_2)](1 - \alpha^2)^{-\frac{1}{2}}),
\\
(h_{\alpha}^{-1}(u_1, u_2) =& \Phi ( \Phi^{-1}(u_1)(1 - \alpha^2)^{\frac{1}{2}} + \alpha \Phi^{-1}(u_2)),
\end{align*}
where $\Phi$ is the CDF of the standard normal distribution.
Now we set
\begin{align*}
\text{LHS} =& h_{\alpha}(\Phi(X_t), \Phi(X_{t-1})
\\
=& \Phi ([\Phi^{-1}(\Phi(X_t)) - \alpha \Phi^{-1}(\Phi(X_{t-1}))](1 - \alpha^2)^{-\frac{1}{2}})
\\
=& \Phi([X_t - \alpha X_{t-1}](1-\alpha^2)^{-\frac{1}{2}}),
\\
\text{RHS} =& \Phi ( \Phi^{-1}(w_t)(1 - \beta^2)^{\frac{1}{2}} + \beta \Phi^{-1}(w_{t-1})).
\end{align*}
Setting LHS=RHS and $\Phi^{-1}(w_t) =: \varepsilon_t, \ \forall t \in \{1,\hdots,T\}$, where due to the quantile transformation $\varepsilon_t \stackrel{iid}{\sim} N(0,1)$, we obtain
\begin{align*}
(& X_t - \alpha X_{t-1}) (1 - \alpha^2)^{-\frac{1}{2}} = \varepsilon_t (1 - \beta^2)^{\frac{1}{2}} + \beta \varepsilon_{t-1},
\\
\iff & X_t - \alpha X_{t-1} = (1-\alpha^2)^{\frac{1}{2}}(1-\beta^2)^{\frac{1}{2}} \varepsilon_t + \beta (1 - \alpha^2)^{\frac{1}{2}} \varepsilon_{t-1},
\\
\iff & X_t - \alpha X_{t-1} = \tilde{\varepsilon}_t + \frac{\beta}{(1-\beta^2)^{\frac{1}{2}}} \tilde{\varepsilon}_{t-1}.
\end{align*}
This is the model equation for the ARMA($1,1$) model with AR-parameter $\alpha$, MA-parameter $\frac{\beta}{(1-\beta^2)^{\frac{1}{2}}}$ and white noise normal innovations with variance \SPP{$(1-\alpha^2)(1-\beta^2)$.}
\item[ii)] The $h$-function associated to the independence copula is $h_{\Pi}(u_1, u_2) = \frac{\partial (u_1 u_2)}{\partial u_2} =u_1$. The corresponding inverse $h$-function is $h_{\Pi}^{-1}(u_1, u_2) = u_1$. It follows that $h_{\theta}^{-1}(w_t,w_{t-1}) = w_t$, so the updating equation becomes $U_t = h_{\phi}^{-1}(w_t, U_{t-1}), w_t \stackrel{iid}{\sim} U(0,1)$, which is the updating equation of the classical copula-based time series model.
\item[iii)] \SPP{We follow the same steps as for ii) and see that $\{U_t\}_{t\in\mathbb{Z}}$ follows the updating equation $U_t = h_{\beta}(w_t, w_{t-1})$, which is the MAG$(1)$ updating equation.}
\item[iv)] Again we use that the inverse $h$-function associated to the independence copula is $h_{\Pi}^{-1}(u_1, u_2) = u_1$. From this it follows that the model updating equation is $U_t = w_t$, where $w_t \stackrel{iid}{\sim} U(0,1)$.
\end{itemize}
\end{proof}
\subsection{Proof of Proposition~\ref{Prop:MAG(1)_Distributional}: Distributional Properties of the MAG$(1)$-Copula Process}
\label{App:Subsection:Proof:Distributional_properties_MAG(1)}
%
%TODO Additional result of joint distribution and stationarity and ergodicity.
\noindent The MAG(1) model is given by
\begin{align*}
U_t = h_{\theta}^{-1}(w_t, w_{t-1}), \ \ w_t \stackrel{iid}{\sim} U(0,1)
\end{align*}
The marginal distribution is standard uniform and the time series is stationary and ergodic. 
\begin{proof}
For $a \in (0,1)$, it holds that
\begin{align*}
\mathbb{P}(U_t \leq a) =& \mathbb{P}(h_{\theta}^{-1}(w_t, w_{t-1}) \leq a),
\\
=& \mathbb{E}( h_{\theta}(a, w_{t-1} )),
\\
=& \int_{0}^1 h_{\theta}(a, w_{t-1}) dw_{t-1} = a.
\end{align*}
For $a \leq 0, \ \mathbb{P}(U_t \leq a) = 0$ and for $a \geq 1, \ \mathbb{P}(U_t \leq a) = 1$ since $U_t \in (0,1)$. So $\mathbb{P}(U_t \leq a)$ is the distribution function of the uniform distribution. This \SPP{result} would also follow directly from \SPP{the consideration} that $U_t$ results from quantile transforming a uniform random variable w.r.t. a quantile function corresponding to a copula, conditioning on an independent uniform random variable. It is well-known that the resulting random variable is uniform. \SPP{Stationarity follows directly because the marginal distribution and the updating equation are time invariant. Ergodicity follows directly from the fact that the process is $1$-dependent.}
%TODO Move the joint distribution
The joint distribution of $(U_t, U_{t-1})$ can be calculated as follows
\begin{align*}
F_{U_t, U_{t-1}}(a, b) =& \mathbb{P}(U_t \leq a, U_{t-1} \leq b)
\\
=& \mathbb{P}(h_{\theta}^{-1}(w_t, w_{t-1}) \leq a, h_{\theta}^{-1}(w_{t-1}, w_{t-2}) \leq b)
\\
=& \mathbb{E}(h_{\theta}(a, w_{t-1}) 1_{\{w_{t-1} \leq h_{\theta}(b, w_{t-2}\}} ) 
\\
=& \int_0^1 \int_0^{h_{\theta}(b, w_{t-2})} h_{\theta}(a, w_{t-1}) dw_{t-1} dw_{t-2}
\\
=& \int_0^1 C_{\theta}[a, h_{\theta}(b, w_{t-2}) ] dw_{t-2}
\end{align*}
Note that implicitly it was used that the marginal distribution of $U_t$ is standard uniform and time independent. Now using that the joint distribution of $(U_t,U_{t-1})$ is time independent, it can be shown that also $F_{U_t,U_{t-1},U_{t-2}}$ is time independent. This procedure can be generalized and then by induction it follows that the process is strictly stationary. (Furthermore it can be seen that $F_{U_t, U_{t-1}}(a, b)$ is increasing in $a$ and $b$, the triangle inequality is fulfilled and $F_{U_t, U_{t-1}}(1, b) = b$ as well as $F_{U_t, U_{t-1}}(a, 1) = a$. Hence $F_{U_t, U_{t-1}}(a, b)$ is a copula.) To establish ergodicity, notice that the process is $1$-dependent. Then ergodicity follows immediately.
\end{proof}
\subsection{Proof of Proposition~\ref{Prop:Convergence_normal_MAGMAR(1,0)}}
\label{App:Subsection:Proof:Convergence_MAGMAR(1,0)}
\begin{proof}
The inverse $h$-function associated to a normal copula with parameter $\phi$ is given as $h_{\phi}^{-1}(u_1,u_2) = \Phi(\Phi^{-1}(u_1)(1 - \phi^2)^{\frac{1}{2}} + \phi \Phi^{-1}(u_2))$ \citep{smith2010modeling}. Then $A_{t,n}$ can be written as
\begin{align*}
A_{t,n} = \Phi\left( (1- \phi^2)^{\frac{1}{2}} \sum_{i = 0}^{n} \phi^i \Phi^{-1}(w_{t-i}) \right).
\end{align*}
Since $w_t \stackrel{iid}{\sim} U[0,1]$ and $\phi^i$ is decreasing geometrically in $i$, the variance of the sum $\sum_{i = 0}^{n} \phi^i \Phi^{-1}(w_{t-i})$ is bounded and it converges almost surely \citep{billingsley1995probability}. The limiting random variable is distributed as $N(0, \frac{1}{1 - \phi^2})$. By the continuous mapping theorem and the continuity of $\Phi$ as well as the continuity of the multiplication with $(1-\phi^2)^{\frac{1}{2}}$ it follows that $A_{t,n}$ also converges almost surely. Since $(1- \phi^2)^{\frac{1}{2}} \sum_{i = 0}^{n} \phi^i \Phi^{-1}(w_{t-i}) \sim N(0,1)$, the limiting random variable is standard uniformly distributed.
\end{proof}
\subsection{Proof of Proposition~\ref{Prop:Stationarity_MAGMAR(1,1)_normal_AR}: Strict Stationarity of MAGMAR$(1,1)$ with normal AR-Copula}
\label{App:Subsection:Proofs:Stationarity_MAGMAR(1,1)_normal_AR}
\begin{proof}
Similar to the reasoning in the proof of \SP{proposition} \ref{Prop:Convergence_normal_MAGMAR(1,0)}, the MAG$(\infty)$ representation of $\{U_t\}$ can be written as
\begin{align*}
A_{t,n} = \Phi\left( (1- \phi^2)^{\frac{1}{2}} \sum_{i = 0}^{n} \phi^i \Phi^{-1}(h_{\theta}^{-1}(w_{t-i}, w_{t-i-1})) \right).
\end{align*}
Now notice that $\{V_t, \ V_t:=h_{\theta}^{-1}(w_t, w_{t-1})\}_{t \in \mathbb{Z}}$ is a MAG$(1)$ time series. From proposition \ref{Prop:MAG(1)_Distributional} it follows that $\{V_t\}_{t \in \mathbb{Z}}$ is strictly stationary, so $\{\Phi^{-1}(V_t)\}_{t \in \mathbb{Z}}$ is also strictly stationary. Since by proposition \ref{Prop:Convergence_normal_MAGMAR(1,0)} converges almost surely and thereby constitutes a measurable mapping, by Theorem 36.4 of \cite{billingsley1995probability} it follows that the process is strictly stationary.
\end{proof}
\subsection{Proof of Proposition~\ref{Prop:AR-Copula_Contraction}: Stationarity of MAGMAR$(1,1)$ based on a Contraction}
\label{App:Subsection:Proofs:AR_Copula_Contraction}
\begin{proof}
\SP{The stationarity and ergodicity of the time series $\{\kappa^{-1}(U_t)\}_{t\in\mathbb{Z}}$ follows directly from theorem 4.40 from \cite{douc2014nonlinear}, since assumptions A4.37 and A4.38 are fulfilled by assumption, and assumption A4.36 is fulfilled since the MAG$(1)$ time series is stationary and ergodic. Since $\kappa$ is a bijective and contemporaneous transformation, $\{U_t\}$ is also stationary.}
\end{proof}
\subsection{Proof of Proposition~\ref{Prop:Conditional_Density_MAGMAR(1,1)}: Conditional distribution of MAGMAR(1,1)}
\label{App:Subsection:Proof:ConditionalDistributionMAGMAR11}
\noindent
Let $g_{\mathbf{U}_{t-1}}:(0,1) \rightarrow (0,1)$ be given as
\begin{align*}
g_{\mathbf{U}_{t-1}}(x) =& h_{\phi}^{-1}(h_{\theta}^{-1}(x, w_{t-1}), U_{t-1}) \quad \text{where} \ U_{t-1}, w_{t-1} \ \text{are deterministic},
\end{align*}
It holds that $U_t|\mathcal{F}_{t-1} = g_{\mathbf{U}_{t-1}}(w_t)|\mathcal{F}_{t-1}$. Application of the density transformation theorem \citep[][Sect. 4.7]{devore2012modern} yields the result. The Jacobian of $g$ can be calculated as
\begin{align*}
\Rightarrow g^{-1}(x) =& h_{\theta}(h_{\phi}(x, U_{t-1}), w_{t-1})),
\\
\Rightarrow |\frac{d}{dx} g^{-1}(x)| =& c_{\theta}(h_{\phi}(x, U_{t-1}), w_{t-1}) c_{\phi}(x, U_{t-1}).
\end{align*}
Using that $w_t \sim U(0,1)$ and hence $f_{w_t}(x) = 1_{(0,1)}(x)$, the resulting conditional density for $U_t|\mathcal{F}_{t-1}$ is
\begin{align*}
f_{U_t|\mathcal{F}_{t-1}}(x) = c_{\theta}(h_{\phi}(x, U_{t-1}), w_{t-1}) c_{\phi}(x, U_{t-1}).
\end{align*}
For the adjusted MAGMAR-copula model, we follow the same procedure. In this case the function $g_{\mathbf{U}_{t-1}}$ is given as
\begin{align}
g_{\tilde{\mathbf{U}}_{t-1}}(x) = \Psi [ h_{\phi}^{-1}(h_{\theta}^{-1}(x, w_{t-1}), \Psi^{-1}(\tilde{U}_{t-1})) ]
\end{align} 
and consequently
\begin{align}
g^{-1}(x) =& h_{\theta}(h_{\phi}(\Psi^{-1}(x), \Psi^{-1}(\tilde{U}_{t-1})), w_{t-1})),
\\
\Rightarrow |\frac{d}{dx} g^{-1}(x)| =& c_{\theta}(h_{\phi}(\Psi^{-1}(x), \Psi^{-1}(\tilde{U}_{t-1})), w_{t-1}) \times c_{\phi}(\Psi^{-1}(x), \Psi^{-1}(\tilde{U}_{t-1})) \times \frac{1}{\psi(\Psi^{-1}(x))}
\end{align}
\subsection{Proof of proposition~\ref{Prop:MAG(1)_Consistency}}
\label{App:Subsection:Proof:MAG(1)_consistency}
\begin{proof}
We make use of the results from \cite{wooldridge1994estimation} for the asymptotic behavior of M-estimators. First we use Theorem 4.2 to establish the uniform law of large numbers for the score-time series, $\{l_t\}$, then we use Theorem 4.3 to establish the consistency of the ML-estimators.
\begin{itemize}
\item[1)] We check the conditions for Theorem 4.2 from \cite{wooldridge1994estimation}:
\begin{itemize}
\item[i)] $\Theta \subset \mathbb{R}^p$ is compact. This is true by assumption.
\item[ii)] $l_t(\theta;u,v) = \log c_{\theta}[u, v]$ is continuous in $u,v$ and $\theta$ and hence measurable.
\item[iii)] \SP{$l_t(\theta; (U_s)_{s\leq t})$ fulfills the WLLN. In order to establish the WLLN, we first establish the ergodicity of $\{\hat{w}_t\}$. Since the time series $\{\hat{w}_t\}$ fulfills assumptions from A4.36-A4.38 from \cite{douc2014nonlinear}, theorem 4.40 from the same reference is applicable. The theorem establishes the stationarity and ergodicity of $\{\hat{w}_t\}$. Additionally, $\{U_t\}$ is $1$-dependent. Then it follows that $\{l_t\}$ is stationary and ergodic and hence fulfills the WLLN.}
%To see this notice that the MAG$(1)$-Copula time series is $1$-dependent. Then $l_t(\theta; (U_s)_{s\leq t})  = \log c_{\theta}[U_{t}, w_{t-1}] = \log c_{\theta}[h_{\theta}^{-1}(w_t, w_{t-1}), w_{t-1}]$ is also $1$-dependent. It follows that $l_t((\theta; (U_s)_{s\leq t})$ is ergodic and by virtue of the ergodic theorem subject to the WLLN.
\item[iv)] There is $c_t((U_s)_{s \leq t})$ such that $|l_t(\theta_1; (U_s)_{s\leq t}) - l_t(\theta_2; (U_s)_{s\leq t})| \leq c_t((U_s)_{s \leq t}) ||\theta_1 - \theta_2||$ and $c_t((U_s)_{s \leq t})$ is subject to the WLLN. To see this notice that by \SPP{i) and ii)} it follows that $l_t(\theta; (U_s)_{s\leq t})$ is Lipschitz-continuous in $\theta$, i.e. there is $L>0$ such that $|l_t(\theta_1; (U_s)_{s\leq t}) - l_t(\theta_2; (U_s)_{s\leq t})| \leq L ||\theta_1 - \theta_2||$. Since $L$ is deterministic it trivially fulfills the WLLN.
\end{itemize}
All assumptions are met, it follows that $l_t(\theta; (U_s)_{s\leq t})$ fulfills the UWLLN.
\item[2)] We use Theorem 4.3 from \cite{wooldridge1994estimation} to establish the consistency of the ML-estimator $\hat{\theta}$. Note that in this set-up, there are no nuisance parameters, so the application of the theorem is simplified. Again, we check the conditions:
\begin{itemize}
\item[i)] $\Theta$ is compact,
\item[ii)] $l_t$ is continuous and measurable.
\item[iii)] $l_t$ fulfills the the UWLLN (as shown above).
\item[iv)] $\theta_0 = \argmin_{\theta \in \Theta}  \mathbb{E}(l_t(\theta))$ is unique by assumption.
\end{itemize}
It follows that the maximum likelihood estimator is consistent.
\end{itemize}
\end{proof}

\section{Additional Background Information}

\subsection{Copulas used in this Work}
\label{App:Subsection:Copulas_in_this_work}
\noindent \textbf{Normal Copula:} The Normal copula models linear dependencies. It is constructed by the inversion method using the multivariate normal distribution \citep{smith2021implicit, nelsen2006introduction},
\begin{align}
C^{\text{normal}}_{\Sigma}[u_1,\hdots,u_p] = \Phi_{\Sigma}(\Phi^{-1}(u_1),\hdots,\Phi^{-1}(u_d)].
\label{NormalCopula}
\end{align}
Here $\Phi_{\Sigma}$ denotes the multivariate normal distribution function with correlation matrix $\Sigma \in [-1,1]^{p \times p}$, symmetric and positive definite. $\Phi^{-1}$ denotes the quantile function of the univariate standard normal distribution. \\
\textbf{$t$-copula:} The $t$-copula models linear dependencies with heavy tails. Similar to the normal copula it is also constructed by the inversion method. It is given as \citep{demarta2005t},
\begin{align}
C^t_{\Sigma, \nu}[u_1,\hdots,u_p] = t_{\Sigma, \nu}(t_{\nu}^{-1}(u_1),\hdots, t_{\nu}^{-1}(u_p)),
\label{tCopula}
\end{align}
where $t_{\Sigma, \nu}$ denotes the multivariate t-distribution with correlation matrix $\Sigma \in [-1,1]^{p \times p}$, symmetric and positive definite, degrees of freedom $\nu > 0$ and $t_{\nu}^{-1}$ is the quantile function of the standard $t$-distribution with degrees-of-freedom $\nu$. For $\nu \rightarrow \infty$, the $t$-copula approaches the normal copula. For small $\nu$, the dependence structure given by the $t$-copula is heavy-tailed, meaning that extreme events coincide. \\
\textbf{Gumbel Copula:} The Gumbel copula is an Archimedean copula \citep[see][]{hofert2008sampling, genest1993statistical}, hence it is constructed via $C[u_1,\hdots,u_p] = \psi(\psi^{-1}(u_1)+\hdots+\psi^{-1}(u_p))$ by choosing a continuous and non-increasing generator function $\psi$. The generator function for the Gumbel copula is given as  $\psi(t) = \exp(-t^{\frac{1}{\rho}})$ \citep{hofert2008sampling}. The resulting copula is 
\begin{align}
C^{\text{Gumbel}}_{\rho} = \exp(-((-\ln(u_1))^\rho + \hdots + (-\ln(u_p))^{\rho})^{\frac{1}{\rho}}).
\label{GumbelCopula}
\end{align}
The Gumbel copula has non-zero upper tail dependence while having no lower tail dependence \citep[][Ex. 5. 22.]{nelsen2006introduction}, hence displaying an asymmetric dependence structure.\\
\textbf{Independence Copula:} The independence copula simply models independence. Accordingly it is given by the product of its arguments \citep{nelsen2006introduction},
\begin{align}
\Pi[u_1, \hdots, u_p] = \prod_{i=1}^{p} u_i.
\end{align}
Choosing $\Pi$ as the copula in Eq.~\ref{Sklar}, it can be seen, that the joint distribution is just the product of marginals. Hence, random variables having the independence copula as their copula implies independence and vice versa. For $\Sigma = I_{p \times p}$, the normal copula becomes the independence copula. For $\Sigma = I_{p \times p}$ and $\nu \rightarrow \infty$ the $t$-copula becomes the independence copula and for $\rho = 1$, the Gumbel copula becomes the independence copula.
\section{Additional Plots}
\label{App:Section:Additional_Plots}
\begin{figure}
\centering
\includegraphics[scale=0.2]{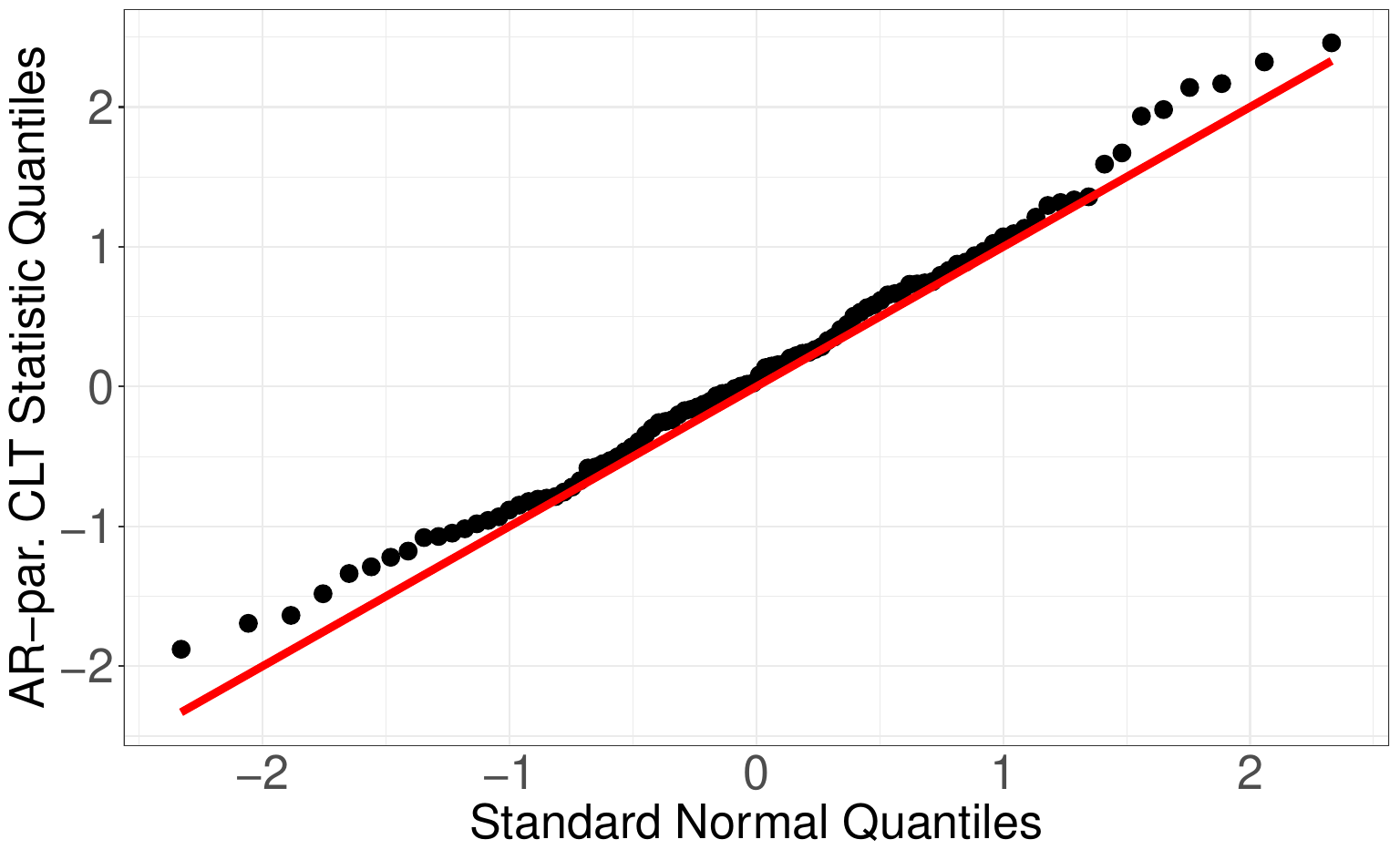}
\includegraphics[scale=0.2]{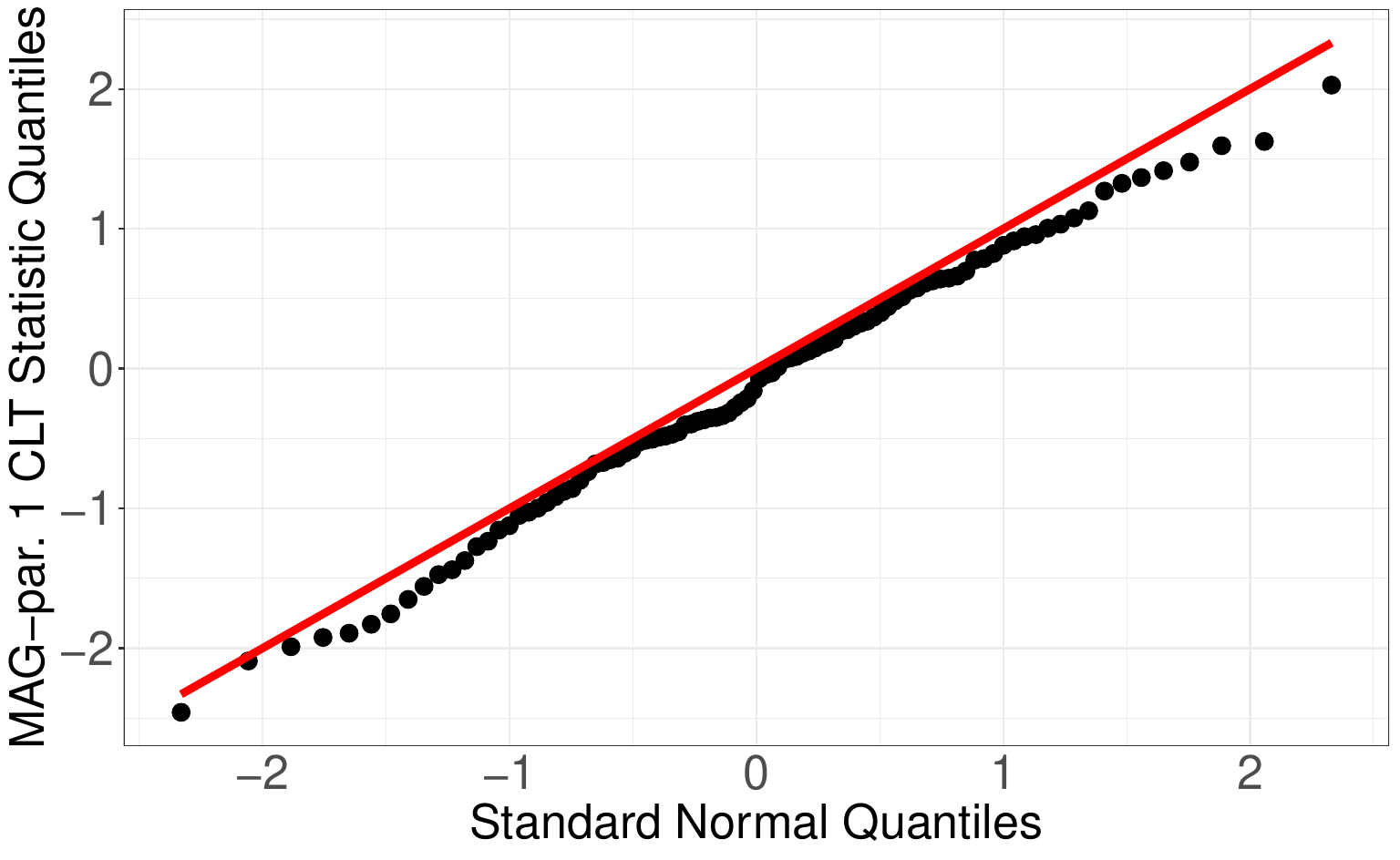}
\includegraphics[scale=0.2]{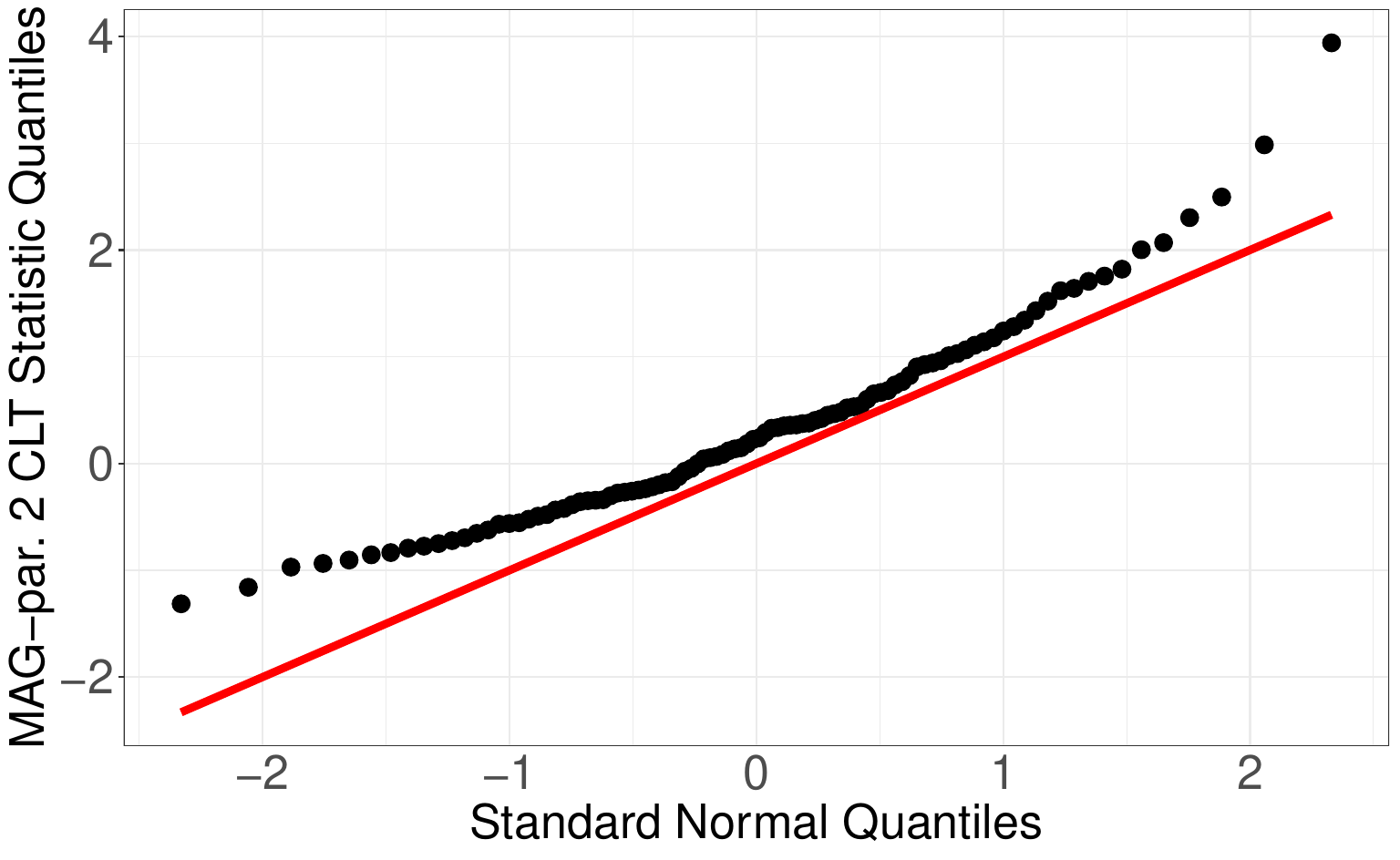}
\caption{QQ-plot of the CLT-statistic of the first MAG-parameter estimator in the $\Psi^{(2)}$-MAGMAR$(1,1)$-(g)-(t) model from the simulation study in Sect.~\ref{Subsection:Simulation_Study:Asymptotics}.}
\label{App:Fig:QQplot:MAG(1)_par1}
\end{figure}

%\begin{figure}
%\centering
%\includegraphics[scale=0.5]{simI_scatterAR1}
%\caption{Scatter plot of consecutive pseudo-observations of the simulated time series.}
%\label{Fig:SimI_scatterAR1}
%\end{figure}

\begin{figure}
\centering
\includegraphics[scale=0.3]{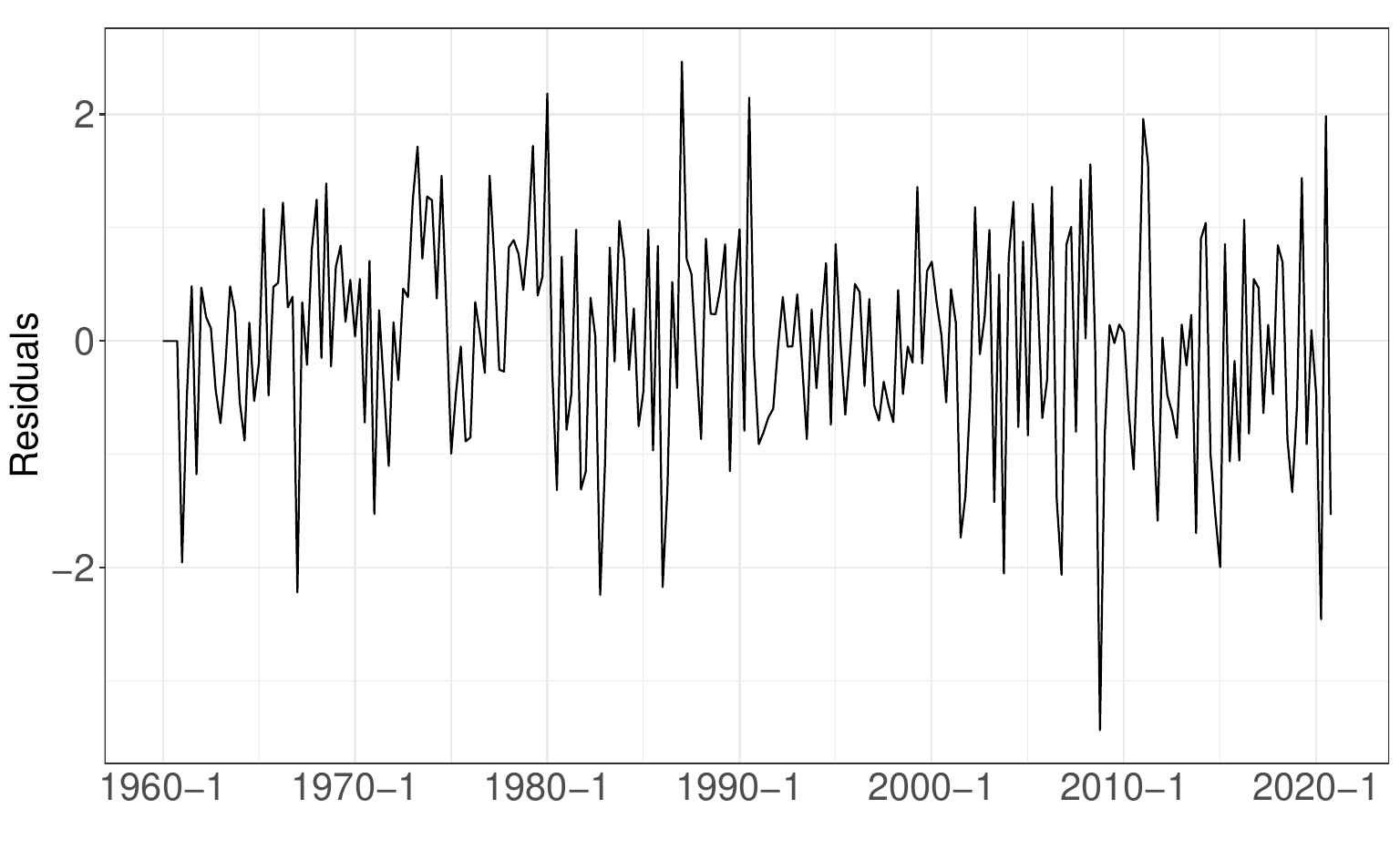}
\includegraphics[scale=0.3]{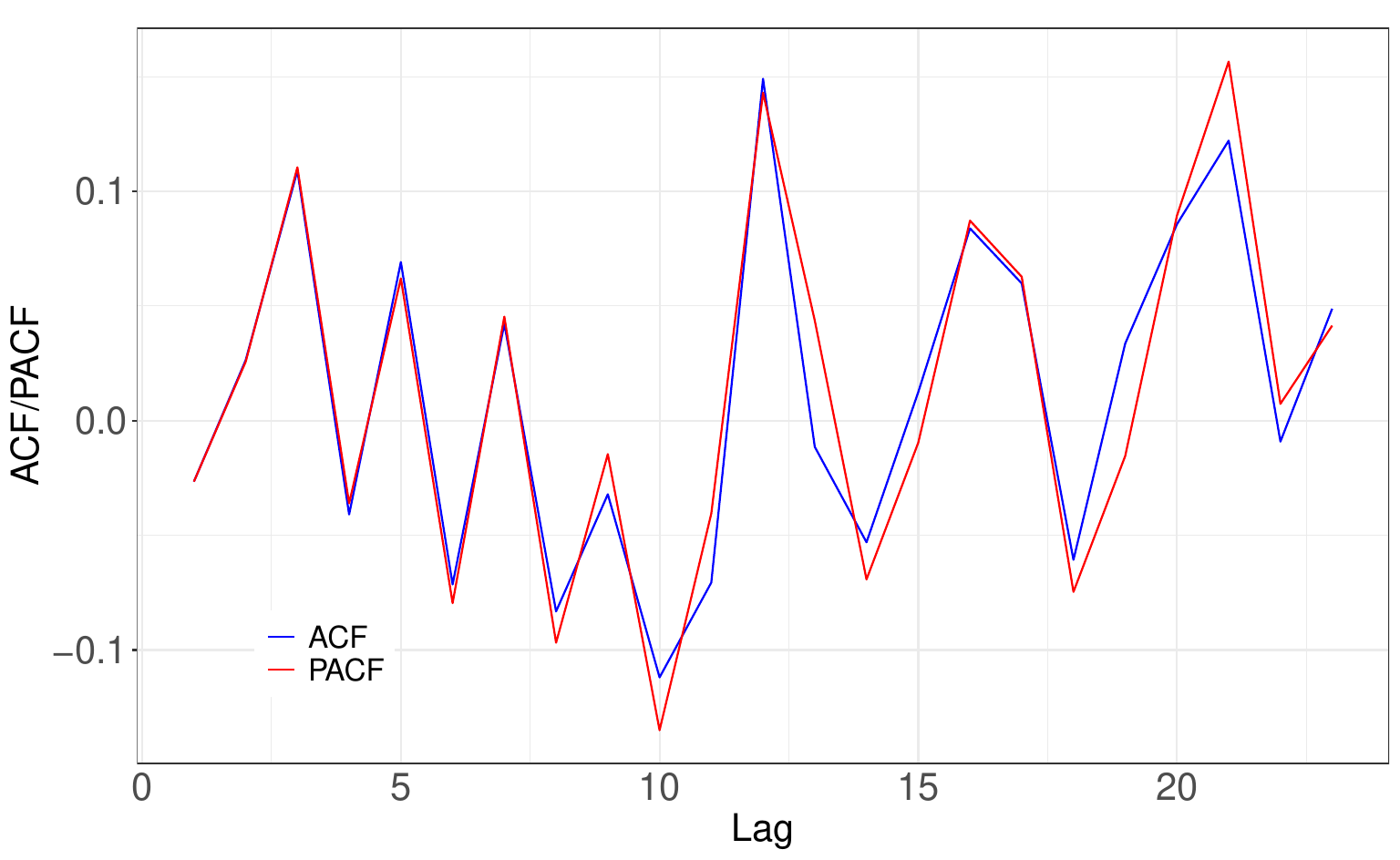}
\caption{Residuals and their ACF/PACF of the $\Psi^{(2)}$-MAGMAR$(4,1)$-(g,i,n,g)-(t) model on US inflation. The residuals were transformed to the $z$-scale by applying the standard normal quantile function.}
\label{App:Fig:Inflation_residual}
\end{figure}

\end{document}